\documentclass[12pt]{article}
\usepackage{amsmath}
\usepackage{amsfonts}
\usepackage{amssymb}
\usepackage{amsthm}
\usepackage{cite}
\usepackage[usenames,dvipsnames]{xcolor}
\usepackage{hyperref}
\usepackage{mathrsfs}
\usepackage{subcaption}
\usepackage{caption}
\DeclareMathOperator*{\slim}{s-lim}
\newlength{\dinwidth}
\newlength{\dinmargin}

\usepackage{tikz}
\usetikzlibrary{trees}
\usetikzlibrary{decorations.pathmorphing}
\usetikzlibrary{decorations.markings}

\newcommand{\bh}{\mathbf{h}}
\newcommand{\bu}{\mathbf{u}}
\newcommand{\msf}{\mathsf}

\newcommand{\bell}{\boldsymbol{\ell} }

\newcommand{\beps}{\bar\eps}
\newcommand{\mrm}{\mathrm}
\renewcommand{\mathbf}{\boldsymbol}

\newcommand{\ext}{\bullet}
\newcommand{\Span}{\mathrm{Span}}
\newcommand{\mcU}{\mathcal U}

\newcommand{\bv}{\mathbf{v}}

\newcommand{\wtla}{\la}
\newcommand{\wtU}{U}
\newcommand{\inc}{\mathrm{in}}
\newcommand{\ba}{\mathbf{a}}

\newcommand{\bxi}{\boldsymbol{\xi}}
\newcommand{\bE}{\boldsymbol{E}}
\newcommand{\Sp}{\mathrm{Sp}}
\newcommand{\hA}{\hat A}
\newcommand{\tout}{\overset{\out}{\times}}
\newcommand{\bla}{\boldsymbol{\xi}}
\newcommand{\bn}{\mathbf{n}}
\newcommand{\by}{\mathbf{y}}
\newcommand{\bp}{\mathbf{p}}
\newcommand{\bP}{\mathbf{P}}
\newcommand{\loc}{\mathrm{loc}}
\newcommand{\mfb}{\mathfrak A}
\newcommand{\bx}{\mathbf{x}}
\newcommand{\res}{\restriction}

\newcommand{\mcF}{\mathcal F}
\newcommand{\mcP}{\mathcal P}
\newcommand{\mcL}{\mathcal L}
\newcommand{\mcK}{\mathcal K}
\newcommand{\ball}{\mathsf{B}}
\newcommand{\hyp}{\mathsf{H}}

\newcommand{\out}{\mathrm{out}}

\newcommand{\mcC}{\mathcal C}
\newcommand{\el}{\mathrm{el}}

\renewcommand{\i}{\mathrm i}

\newcommand{\wlim}{\mathrm{w-}\lim}

\newcommand{\cc}{\mathrm{c} }
\newcommand{\wt}{\widetilde}

\newcommand{\Om}{\Omega}
\newcommand{\ga}{\gamma}

\newcommand{\La}{\Lambda}

\newcommand{\be}{\beta}

\newcommand{\pa}{\partial}

\newcommand{\Ran}{\mathrm{Ran}}
\newcommand{\ov}{\overline}

\newcommand{\mfh}{\mathfrak{h}}

\newcommand{\eps}{\varepsilon}
\newcommand{\de}{\delta}

\newcommand{\De}{\Delta}
\newcommand{\e}{\mathrm{e}}

\newcommand{\pho}{\mathrm{ph}}

\newcommand{\nin}{\noindent}
\newcommand{\si}{\sigma}
\newcommand{\ph}{\phantom}
\newcommand{\h}{\fr{1}{2}}
\newcommand{\nat}{\mathbb{N}}
\newcommand{\hil}{\mathcal{H}}
\newcommand{\om}{\omega}
\newcommand{\mfa}{\mathfrak{A}}
\newcommand{\mco}{\mathcal{O}}

\newcommand{\supp}{\mathrm{supp}}
\newcommand{\fr}[2]{\frac{#1}{#2}}
\newcommand{\al}{\alpha}
\newcommand{\real}{\mathbb{R}}
\newcommand{\complex}{\mathbb{C}}
\newcommand{\la}{\lambda}
\newcommand{\non}{\nonumber}
\newcommand{\Ga}{\Gamma}

\newcommand{\lan}{\langle}
\newcommand{\ran}{\rangle}

\def\proof{\noindent{\bf Proof. }}
\def\qed{$\Box$\medskip}

\newtheorem{theoreme}{Theorem } [section]
\newtheorem{proposition}[theoreme]{Proposition}
\newtheorem{lemma}[theoreme]{Lemma}
\newtheorem{definition}[theoreme]{Definition}
\newtheorem{corollary}[theoreme]{Corollary}
\newtheorem{remark}[theoreme]{Remark}
\newtheorem{example}[theoreme]{Example}
\newtheorem{criterion}[theoreme]{Criterion}
\newtheorem{conjecture}{Conjecture}
\newtheorem{assumption}{Assumption}

\newcommand{\bea}{\begin{assumption}}
	\newcommand{\eea}{\end{assumption}}
\newcommand{\beco}{\begin{conjecture} }
	\newcommand{\eeco}{\end{conjecture} }
\newcommand{\beq}{\begin{equation}}
	\newcommand{\eeq}{\end{equation}}
\newcommand{\beqa}{\begin{eqnarray}}
	\newcommand{\eeqa}{\end{eqnarray}}
\newcommand{\ben}{\begin{arabicenumerate}}
	\newcommand{\een}{\end{arabicenumerate}}
\newcommand{\bex}{\begin{example}}
	\newcommand{\eex}{\end{example}}
\newcommand{\ber}{\begin{remark}}
	\newcommand{\eer}{\end{remark}}
\newcommand{\bec}{\begin{corollary}}
	\newcommand{\eec}{\end{corollary}}
\newcommand{\bep}{\begin{proposition}}
	\newcommand{\eep}{\end{proposition}}
\newcommand{\becr}{\begin{criterion}}
	\newcommand{\eecr}{\end{criterion}}

\def\bel{\begin{lemma}}
	\def\eel{\end{lemma}}
\def\bet{\begin{theoreme}}
	\def\eet{\end{theoreme}}
\def\bed{\begin{definition}}
	\def\eed{\end{definition}}

\begin{document}
\title{Compton scattering in the Buchholz-Roberts framework of relativistic QED.  } 

\author{Sabina Alazzawi\footnote{E-mail: sabina.alazzawi@tum.de} \ and Wojciech Dybalski\footnote{E-mail: dybalski@ma.tum.de}  \\\\
Zentrum Mathematik, Technische Universit\"at M\"unchen,\\
D-85747 Garching, Germany   }

\date{}

\maketitle

\vspace{-0.7cm}
\begin{center}
\small{\emph{Dedicated to the memory of John E. Roberts}}
\end{center}
\vspace{0.1cm}
\begin{abstract}

We consider a Haag-Kastler net in a positive energy representation,
admitting massive Wigner particles and asymptotic fields of massless bosons. We show that states of the massive particles are always vacua of the massless asymptotic fields. Our argument is  based on the Mean Ergodic Theorem in a certain extended Hilbert space.  
As an application of this result we construct the outgoing isometric  
wave operator for Compton scattering in QED in a class of representations recently proposed by Buchholz and Roberts. 
In the course of this analysis we use our new technique to further simplify scattering theory of massless bosons in the vacuum sector. A general discussion of the status of the infrared problem in the setting of Buchholz and Roberts is given.  
	
\end{abstract}

\section{Introduction}
\setcounter{equation}{0}

In general, the term \emph{infrared problems} can be understood as complications  in mathematical description  of quantum systems encountered at 
large spatio-temporal scales.  However, its conventional definition is more specific
and  refers to difficulties in scattering theory of such systems 
in the presence of  long range forces and/or massless particles. 
The simplest and well understood example    
is Coulomb scattering in quantum mechanics which requires the Dollard modifications of  the wave operators. Infrared problems
in quantum electrodynamics (QED) still
evade a satisfactory solution and constitute an active field of research in mathematical physics. 
Among many advances of  recent years \cite{MS15, He14.0, CFP07,BR14},  a particularly radical proposal 
was put forward by Buchholz and Roberts in the setting of algebraic quantum field theory (AQFT) \cite{BR14}.
In essence, these authors suggest that
after restricting attention to measurements in some future lightcone $V$, 
infrared problems should disappear. Buchholz and Roberts adopt the
general point of view on  infrared problems and illustrate their ideas by 
results on superselection structure of QED.
However, conventional infrared problems, understood as complications in scattering theory, are not treated in their work. 
It is therefore an open question if the appealing ideas of Buchholz and Roberts are helpful for analysis of collision processes in QED. We give a partial answer in this work.

Infrared problems in QED can be traced back to the fact that the spacelike asymptotic flux of the electric field 
\beqa
\phi(\bn)=\lim_{r\to\infty} r^2\bn \bE(r\bn), \quad \bn\in S^2
\label{fluxes}
\eeqa
commutes with all local observables \cite{Bu82}. Since this flux is an arbitrary function on the unit sphere $S^2$, restricted only by the Gauss Law,  each value of the electric charge corresponds to uncountably many disjoint irreducible representations of the algebra of observables, which are of potential
physical interest.  This invalidates the standard Doplicher-Haag-Roberts (DHR) 
theory of superselection sectors. For non-zero charges none of these representations can be Poincar\'e covariant, since the existence of $\phi$
is not consistent with unitary action of Lorentz transformations.
For similar reasons,
charged particles cannot have sharp masses \cite{Bu86}. 
This latter difficulty, called the \emph{infraparticle problem}, invalidates the conventional Haag-Ruelle or Lehmann-Symanzik-Zimmermann (LSZ) scattering theory for electrically charged particles. In this situation a charged particle is a
composite object involving  a soft-photon cloud correlated with the particle's velocity. The cloud  is needed for the purpose of `fine-tuning the flux', that is, keeping it constant along the time evolution \cite{Bu82}. Such \emph{infraparticles} have  in fact been constructed in  concrete models of non-relativistic QED \cite{CFP07}.

The above discussion involves a tacit restriction to representations of the algebra of observables of QED in which the flux (\ref{fluxes})
exists.  Buchholz and Roberts consider instead a class of representations in which this is not the case, i.e. the fluctuations of
the electric field tend to infinity under large spacelike translations. 
Thinking heuristically, one way to achieve this is to include highly fluctuating background radiation, emitted in very distant past.
Such radiation, which should not be confused with soft photon clouds mentioned above,  will `blur the flux', that is prevent the existence of the limit in (\ref{fluxes}).  On the other hand, it is clear from Figure~\ref{fig:sub1} and the Huygens principle that this background radiation will stay outside any future lightcone~$V$. 
Thus, inside $V$ one can follow the usual DHR strategy to pass from the defining vacuum representation $\iota$ of the algebra of observables $\mfa$ to an electrically charged positive energy representation $\pi$. To this end, consider a  pair of opposite charges  in a hypercone $\mcC\subset V$, which is a region depicted in  Figure~\ref{fig:sub1} and defined precisely in Subsection~\ref{rep-subsection}.
Next,  transport one of the charges to  lightlike infinity.  As argued in \cite{BR14}, this process of charge creation in $\mcC$ should be only weakly correlated with operations performed in the spacelike complement of $\mcC$ in $V$, denoted $\mcC^\cc$. Therefore, the resulting charged representation $\pi$ should satisfy the following property of \emph{hypercone localization}
\beqa
\pi\res \mfa(\mcC^\cc)\simeq\iota \res \mfa(\mcC^\cc), 
\label{intro-hypercone}
\eeqa
where $\simeq$ denotes unitary equivalence and $\mfa(\mcC^\cc)$ is the algebra of all observables measurable in $\mcC^\cc$. Since $\mcC^\cc\subset V$, this
property is consistent with high fluctuations of the electric field at spacelike infinity, blurring the flux (\ref{fluxes}). (See again Figure~\ref{fig:sub1}). 
As $\phi$ does not exist,
we may require that $\pi$ is
covariant under Poincar\'e transformations and that charged particles 
have sharp masses\footnote{Poincar\'e covariance is used in \cite{BR14} 
at a technical level.  The possibility of sharp masses of charged particles is only mentioned as a problem for future investigations. }.  We adopt these  assumptions in this work and study their consequences.

The  problem of verifying these assumptions in  some concrete models of QED  is outside  the scope of this work.
However,  the above discussion reveals certain similarity of the Buchholz-Roberts ideas to the concept of
\emph{infravacua} \cite[p.59]{Bu82}\cite{Ku98, Kr82}. We recall that such states result from adding to the vacuum a sufficiently strong background field.
Infravacua were constructed in  QED in the external current approximation by Kraus, Polley and Reents \cite{KPR77}.
We believe that a similar analysis in more realistic theories, e.g. translation invariant models of non-relativistic QED, 
could bring interesting new insights into the nature of electrically charged particles.

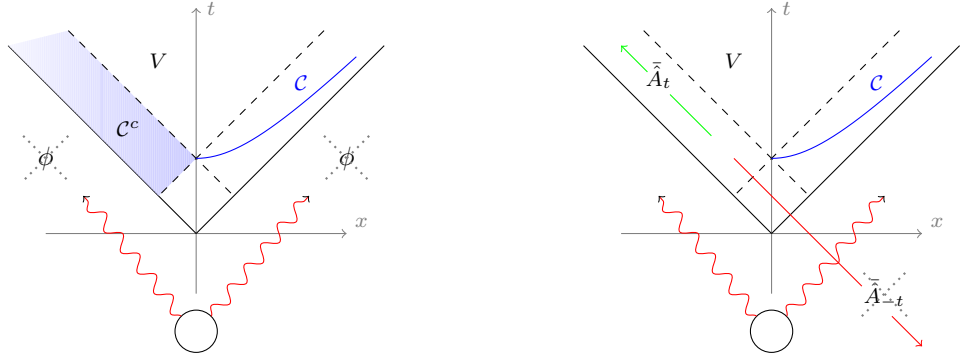
\begin{figure}
\centering
\begin{subfigure}{.5\textwidth}
  \centering
    \begin{tikzpicture}[scale=1]
                \begin{scope}
                \draw[fill] (-2,1)  node{\footnotesize{$\phi$}};
                \end{scope} 
        \begin{scope}
                \draw[fill] (2,1) node{\footnotesize{$\phi$}};
        \begin{scope}[->]
            \draw[gray] (-2,0) -- (2,0) node[anchor=north] {};
            \draw[gray] (0,-.8) -- (0,3) node[anchor=east] {};
        \end{scope}
        \draw (-2.5,2.5) -- (0,0);
        \draw (0,0) -- (2.5,2.5);
       \end{scope}
       \draw[fill,gray] (2.2,.15) node{\scriptsize{$x$}};
       \draw[fill,gray] (.2,3) node{\scriptsize{$t$}};
       \draw[fill] (-.5,2.3) node{\scriptsize{$V$}};
       \shade[left color=blue!5!white,right color=blue!30!white,opacity=0.5] (0,1)--(-1.7,2.7)--(-2.5,2.5) -- (-.5,.5);
       \draw[fill] (-.9,1.4) node{\scriptsize{$\mathcal{C}^c$}};
       \draw[dashed](-1.7,2.7)--(0.5,.5);
       \draw[dashed](1.7,2.7)--(-.5,.5);
       \draw[blue] plot[variable=\t,domain=0:1.5] ({sinh(\t)},{cosh(\t)});
       \draw[fill, blue] (1.4,2) node{\scriptsize{$\mathcal{
       C}$}};
       \draw (0,-1.3) circle (8pt);
       \draw[dotted, thick,gray] (-2.3,1.3)--(-1.7,.7);
       \draw[dotted, thick,gray] (-1.7,1.3)--(-2.3,.7);
       \draw[dotted, thick,gray] (2.3,1.3)--(1.7,.7);
       \draw[dotted, thick,gray] (1.7,1.3)--(2.3,.7);
       \begin{scope}[->]
       \draw[style={decorate, decoration={snake}, draw=red}] (.2,-1.1)--(1.5,.5)node[anchor=north] {};
              \draw[style={decorate, decoration={snake}, draw=red}] (-.2,-1.1)--(-1.5,.5)node[anchor=north] {};
              \end{scope}
      \end{tikzpicture}
  \caption{\footnotesize{A hypercone~localized representation.}}
  \label{fig:sub1}
\end{subfigure}%
\begin{subfigure}{.5\textwidth}
  \centering    
  \begin{tikzpicture}[scale=1]
          \begin{scope}
                  \draw[fill] (1.5,-.8) node{\scriptsize{$\bar{\hat{A}}_{-t}$}};
          \begin{scope}[->]
              \draw[gray] (-2,0) -- (2,0) node[anchor=north] {};
              \draw[gray] (0,-.8) -- (0,3) node[anchor=east] {};
          \end{scope}
          \draw (-2.5,2.5) -- (0,0);
          \draw (0,0) -- (2.5,2.5);
         \end{scope}
         \draw[fill,gray] (2.2,.15) node{\scriptsize{$x$}};
         \draw[fill,gray] (.2,3) node{\scriptsize{$t$}};
         \draw[fill] (-.5,2.3) node{\scriptsize{$V$}};
         \draw[fill] (-1.5,2.1) node{\scriptsize{$\bar{\hat{A}}_t$}};
         \draw[dashed](-1.7,2.7)--(0.5,.5);
         \draw[dashed](1.7,2.7)--(-.5,.5);
         \draw[blue] plot[variable=\t,domain=0:1.5] ({sinh(\t)},{cosh(\t)});
         \draw[fill, blue] (1.4,2) node{\scriptsize{$\mathcal{C}$}};
         \draw (0,-1.3) circle (8pt);
         \draw[dotted, thick, gray] (1.8,-1.1)--(1.2,-.5);
         \draw[dotted, thick,gray] (1.2,-1.1)--(1.8,-.5);
         \begin{scope}[->]
         \draw[style={decorate, decoration={snake}, draw=red}] (.2,-1.1)--(1.5,.5)node[anchor=north] {};
         \draw[style={decorate, decoration={snake}, draw=red}] (-.2,-1.1)--(-1.5,.5)node[anchor=north] {};
         \draw[green] (-1.65,2.15)--(-2,2.5);
         \draw[red](1.6,-1.1)--(2,-1.5)node[anchor=south] {};
         \end{scope}
         \draw[green] (-1.35,1.85)--(-.8,1.3);
         \draw[red] (-.5,1)--(1.2,-.7);
        \end{tikzpicture}
  \caption{\footnotesize{(Non-)existence of asymptotic fields.}}
  \label{fig:sub2}
\end{subfigure}
\caption*{\footnotesize{Figure 1. (a) A hypercone localized representation of QED is equivalent to the vacuum representation
in the causal complement $\mcC^\cc\subset V$ of any hypercone $\mcC\subset V$.
This condition is consistent with the presence of highly fluctuating background radiation emitted in distant past, which is
needed to blur the flux~$\phi$.  (b) If the approximating sequence $[1,\infty) \ni t\mapsto\bar{\hat{A}}_t$ of the \emph{outgoing} asymptotic photon field  
is localized in  $\mcC^\cc$, the existence of the limit $\hat{A}^{\out}$ can be inferred
from the corresponding result in the vacuum representation \cite{Bu77}. However, the \emph{incoming} asymptotic field is
not expected to exist, since its approximating sequence $[1,\infty) \ni t\mapsto\bar{\hat{A}}_{-t}$ collides with the background
radiation.}}
\label{mass-shell-fig}
\end{figure}

\normalsize{}

\vspace{0.5cm}

\nin\textbf{Results.} Let us now give an outline of our results in somewhat simplified terms.
As mentioned above, we consider a Haag-Kastler theory $(\mfa, U)$ in a vacuum representation,
given by the algebra of observables $\mfa\subset B(\hil)$
and a unitary representation of the covering group of the Poincar\'e group
$\wt{P}_+^{\uparrow}\ni \la\mapsto U(\la)$. 
We also consider a  Poincar\'e covariant, positive energy representation $\pi$ of $\mfa$, satisfying the property of hypercone localization~(\ref{intro-hypercone}), which gives rise to a new Haag-Kastler theory  $(\hat\mfa, \hat U)$ on a Hilbert space $\hat \hil$. The vacuum representation is assumed 
to contain massless Wigner particles (`photons') and the representation $\pi$ to contain massive 
Wigner particles (`electrons'). That is, there is a subspace 
$\mfh_{\pho}\subset\hil$ on which $U$ acts as a representation of $\wt{P}_+^{\uparrow}$ with mass $m_{\pho}=0$ and, similarly, 
a subspace  $\hat\mfh_{\el}\subset \hat\hil$ on which $\hat U$ acts as a representation
of  mass $m_{\el}>0$. 

Our goal is to describe Compton scattering, i.e. collision processes involving one electron and some finite number of photons. To
be able to add photons to vectors $\Psi_{\el}\in \hat\mfh_{\el}$ describing one electron, we introduce asymptotic fields of photons via the LSZ prescription. For this purpose, let $\hA\in \hat\mfa$ be a suitable local operator and $\hA(t,\bx):= \hat U(t,\bx)\hA  \hat U(t,\bx)^*$  its spacetime translations. Moreover, let 
$f$ be a solution of the wave equation with compactly supported initial
data, i.e.
\beqa
f(t,\bx)=(2\pi)^{-3/2}\int d^3p\, \e^{\i\bp\bx}\big(\e^{-\i |\bp| t}\wt f_{\mrm{p}}(\bp)+
\e^{\i |\bp| t}\wt f_{\mrm{n}}(\bp)\big), \label{wave-packet}
\eeqa
where $\wt f_{\mrm{p}}(\bp)=\wt f_1(\bp)- \i |\bp|\wt f_2(\bp)$,  $\wt f_{\mrm{n}}(\bp)=\wt f_1(\bp)+\i |\bp|\wt f_2(\bp)$,
$f_1, f_2\in C_0^{\infty}(\real^3)$, determine the positive and negative energy parts of $f$. The asymptotic photon field approximants, given by
\beqa
\bar{\hA}_t:=\fr{1}{\ln\, t}\int_t^{t+\ln\, t} dt'  \int d^3x\, \hA(t',\bx) f(t',\bx), \label{intro-as-field}
\eeqa
finally, give rise to the asymptotic fields via
\beqa
\hA^{\out}:=\lim_{t\to\infty}\bar{\hA}_t \label{intro-limit}.
\eeqa
We shall show that these fields exist as strong limits on the domain $D_{\hat H}\subset \hat\hil$ of vectors of polynomially bounded energy and leave this domain invariant. 
The first step of the proof is inspired by \cite{Bu82}. Namely, we decompose $\bar{\hA}_t$ into a finite number of terms $\bar{\hA}_{i,t}$, $i=1, \ldots, N$,  which are localized in causal complements of some hypercones $\mcC_i$. Then, we use the hypercone localization property
(\ref{intro-hypercone}) and the existence of asymptotic photon fields
in the vacuum representation \cite{Bu77} to obtain the limits $\hA^{\out}_{i}$ on some domains $D_i$. Finally, we use the energy bounds \cite{Bu90, He14}
\beqa
\sup_{t\geq 1} \|\bar{\hA}_t(1+\hat H)^{-1}\|<\infty, \label{intro-energy-bounds}
\eeqa 
where $\hat H$ is the Hamiltonian in representation $\pi$, to obtain
the limits $\hA^{\out}_{i}$ on a common domain $D_{\hat H}$ on which they can be added up to $\hA^{\out}$. (Cf. Figure~\ref{fig:sub2}).

Given  $\hA^{\out}$ we define asymptotic creation and annihilation operators as 
follows
\beqa
\hA^{\out+}:=\int d^4x\, \hA^{\out}(x)\eta(x),\quad  \hA^{\out-}:= (\hA^{\out+})^*, \label{asymptotic-field}
\eeqa
where the Fourier transform $\wt \eta\in C_0^{\infty}(\real^4)$ of $\eta$ is supported outside of
the backward lightcone in energy-momentum space. Since $x\mapsto \hA^{\out}(x)$ is a solution of the wave equation,
this smearing operation restricts the energy transfer of $\hA^{\out}$ to positive values. Summing up, vectors of the form
\beqa
\Psi^{\out}:=\hA_1^{\out+} \ldots  \hA_n^{\out+}\Psi_{\el} \label{Compton-vectors}
\eeqa
are natural candidates for Compton scattering states describing $n$ photons and one electron.

These states can now be used to construct the outgoing wave operator 
\beqa
W^{\out}: \Ga(\mfh_{\pho})\otimes \hat\mfh_{\el}\to \hat\hil,
\eeqa
with $\Ga(\mfh_{\pho})$ being the symmetric Fock space over  $\mfh_{\pho}$. 
$W^{\out}$ maps any configuration of one electron and $n$ independent photons into 
the corresponding vector of the form (\ref{Compton-vectors}). 
However, to show that $W^{\out}$ is well defined and isometric, two 
ingredients are needed. Firstly, the asymptotic creation and annihilation operators $\hA^{\out\pm}$ 
must satisfy the standard canonical commutation relations.  
This can be shown by adapting results from \cite{Bu77, Bu82} to a new geometric situation.
Secondly, single-electron states must play a role of vacua of  the asymptotic photon fields, i.e.
\beqa
\hA^{\out-}\Psi_{\el}=0. \label{annihilation}
\eeqa  
Our proof of this fact, which is the main new technical result of this paper, is outlined below in this
introduction. This proof relies only on the Haag-Kastler postulates.
In particular, the hypercone localization of $\pi$ is not needed to show (\ref{annihilation}).

To formalize the idea that single-electron states are vacua of the asymptotic photon fields, we construct
the corresponding Haag-Kastler theory $(\hat\mfa^{\out}, \hat U)$. More precisely, for any double cone $\mco$ we
define the corresponding local algebra
\beqa
\hat\mfa^{\out}(\mco):=\{\, \e^{\i\hat A^{\out}}\,|\,\bar{\hA}_t=(\bar{\hA}_t)^* \textrm{ for all $t\geq 1$, } 
\ \bar{\hA}_t\in \hat\mfa(\mco)
\textrm{ for small $t\geq 1$ } \}''.
\eeqa  
This definition requires the self-adjointness of $\hat A^{\out}$  resulting from  self-adjoint approximating sequences. 
We show this using  the Nelson commutator theorem \cite{RS2} with the energy bounds~(\ref{intro-energy-bounds}) as an input.
Results from \cite{Fr77}, with the same input, yield Weyl relations for operators of the form $\e^{\i\hat A^{\out}}$. With
this information at hand and relation~(\ref{annihilation})  we verify that states of the form $\om_{\el}(\,\cdot\,):=\lan\Psi_{\el}, \,\cdot \, \Psi_{\el}\ran$,
$\Psi_{\el}\in \hat\mfh_{\el}$, $\|\Psi_{\el}\|=1$, induce vacuum representations of $(\hat\mfa^{\out}, \hat U)$. We point
out that the improvements of the energy bounds made in \cite{He14} (lower powers of the resolvent of $\hat H$ than in \cite{Bu90})  
are important for this part of our analysis.

\vspace{0.5cm}

\nin\textbf{Outline of the proof of (\ref{annihilation}).}
Let $\hat\hil_{\mathrm{c}}\subset \hat\hil$ be the subspace of vectors of bounded energy and $B$ be almost local
operator whose energy-momentum transfer is outside of the future lightcone (cf. Subsections~\ref{almost-local-subs}, \ref{Arveson-subs}). 
Next, we define auxiliary maps $a_B$  introduced in \cite{DG14} by C. G\'erard and one of the present authors, namely
\beqa
& &a_B:\hat\hil_{\mathrm{c}}\to \hat\hil\otimes L^2(\real^3),\\
& &(a_B\Psi)(\bx)=B(\bx)\Psi.
\eeqa  
It is not obvious that the range of $a_B$ is in $\hat\hil\otimes L^2(\real^3)$, but it follows from \cite[Lemma~2.2]{Bu90}, restated 
as Lemma~\ref{HA-lemma}
below. It is easy to see that this map  has the following property
\beqa
a_B \circ f(\hat\bP)=f(\hat\bP+D_{\bx})\circ a_B,
\eeqa
where $f$ is a bounded Borel function, $\hat\bP$ is the momentum operator, $D_{\bx}=-\i\nabla_{\bx}$ and we use the short-hand notation $\hat\bP+D_{\bx}:=\hat\bP\otimes 1_{L^2(\real^3)}+1_{\hat\hil}\otimes D_{\bx}$. 

Let $(\bar{\hat{A}}_t(\eta))^*$, $t\geq 1$, be the approximating sequence  of the asymptotic annihilation operator $\hA^{\out-}$ and put
\beqa
\bar{\hA}_t(\eta)=\bar{\hA}_{t,\mrm{p}}(\eta)+\bar{\hA}_{t,\mrm{n}}(\eta),
\eeqa
where $\bar \hA_{t,\mrm{p}/\mrm{n} }$ involve the positive and negative energy parts of the wave packet (\ref{wave-packet}).
For the positive energy part we have
\beqa
(\bar A_{t,\mrm{p}}(\eta))^*\Psi_{\el}&=&
\fr{1}{\ln\, t}\int_t^{t+\ln\, t} dt'  (1_{\hat\hil}\otimes \lan f_{\mrm{p}} |) (\e^{\i \hat H t' }\otimes  
\e^{\i |D_{\bx}|t'} )a_B\e^{-\i\om_{m_\el}(\hat\bP)t'}\Psi_{\el}\non\\
&=&(1_{\hat\hil}\otimes \lan f_{\mrm{p}}| )\fr{1}{\ln\, t}\int_t^{t+\ln\, t} dt' 
\e^{\i (\hat H+  |D_{\bx}|-\om_{m_{\el}}(\hat\bP+D_{\bx}) )t' }a_B\Psi_{\el}, 
\eeqa
where $B:=A^*(\ov{\eta})$, $\om_{m_{\el}}(k)=\sqrt{k^2+m_{\el}^2}$ and the map
 $(1_{\hat\hil}\otimes \lan f_{\mrm{p}} | ): \hat\hil\otimes L^2(\real^3)\to \hat\hil$ acts according to
\beqa
(1_{\hat\hil}\otimes \lan f_{\mrm{p} }  | )\Phi=\int d^3x \bar{f}_{\mrm{p}}(\bx)\Phi(\bx).
\eeqa
Now we are in position to apply the Mean Ergodic Theorem in $\hat\hil\otimes L^2(\real^3)$, which gives
\beqa
\lim_{t\to\infty}(\bar A_{t,\mrm{p}}(\eta))^*\Psi_{\el}=(1_{\hat\hil}\otimes \lan f_{\mrm{p}}| )F(\{0\}) a_B\Psi_{\el},
\eeqa 
where $F$ is the spectral measure of  $\hat H+  |D_{\bx}|-\om_{m_{\el}}(\hat\bP+D_{\bx})$. Since $\hat H, \hat\bP, D_{\bx}$ commute, spectral
calculus and covariance under Lorentz transformations can be used to show $F(\{0\})=0$. The analysis of $(\bar A_{t,\mathrm{n}  }(\eta))^*\Psi_{\el}$ is analogous.

Apart from verifying (\ref{annihilation}), the technique described above serves as a tool to simplify scattering theory
of massless bosons in the vacuum sector. In particular, the proof of the fact that 
$A^{\out\pm}$ satisfy canonical commutation relations can now  be
accomplished via a Pohlmeyer argument, without referring to the quadratic
decay of the vacuum correlations of local observables. Thus, with the a priori information from \cite{Bu90, He14}  and the present paper, collision theory for massless bosons can be developed in a way completely parallel to the fermionic case \cite{Bu75}.

Since the argument above  does not rely on strict locality, it may also be useful outside of the Haag-Kastler setting, e.g. in theories satisfying
some kind of asymptotic abelianess in spacelike directions.  For example,
it should help to remove Assumption~4 of \cite{He14.0} and Assumption~3 of \cite{DH15}. It might also find applications
in scattering theory of quantum spin systems satisfying the Lieb-Robinson bounds \cite{BDN14}.

\vspace{0.5cm}

\nin\textbf{Discussion.} Let us now  turn to the  status of the infrared problem
in the Buchholz-Roberts setting of relativistic QED. It may seem that by an analogous construction
as above one could obtain the incoming wave operator $W^{\inc}$ as well and define the scattering matrix $S$
of the Compton scattering process in the usual way, namely by putting $S=(W^\out)^*W^\inc$. 
Unfortunately, the situation is  less satisfactory than that. As far as we can
see, the incoming wave operator $W^\inc$ is not available in a
representation $\pi$ which is hypercone localized in a future lightcone. While the hypercone localization 
property~(\ref{intro-hypercone})  allows us  to establish the existence of  the \emph{outgoing} asymptotic photon fields, 
as explained below formula~(\ref{intro-limit}),  it is of no help for the \emph{incoming} photon fields. This is due to the fact that 
the approximating sequences 
of the incoming asymptotic photon fields are localized in regions moving to infinity in negative 
lightlike directions. Heuristically speaking, such regions inevitably collide with the highly fluctuating background radiation,
emitted in very distant past, which must be present in $\pi$ to prevent the  existence of the flux~(\ref{fluxes}). 
It is therefore reasonable to expect that also the incoming  asymptotic photon fields are blurred by this radiation as depicted in   Figure~\ref{fig:sub2}.
As a possible way out, one could consider a representation $\pi'$ hypercone localized in a \emph{backward} lightcone
in which by obvious modifications of our discussion $W^{\inc}$ exists but $W^{\out}$ may not exist. However, the existence of the
scattering matrix remains questionable, since there is no reason to expect that $\pi$ and $\pi'$ are unitarily equivalent.

Like in the conventional approach,  in the absence of the scattering matrix one may try to construct
inclusive collision cross sections. This idea, implemented in AQFT by Buchholz, Porrmann and Stein \cite{BPS91}, amounts in our
situation to preparation of incoming states using asymptotic observables of the form
\beqa
C_t:=\int d^3x\, h(\bx/t) (B^*B)(t,\bx).
\eeqa
Here $h\in C_0^{\infty}(\real^3)$ is supported on velocities of the desired particle and $B$ is an almost-local observable
whose energy-momentum transfer is outside of the future lightcone. 
Due to this latter property,
which cannot be imposed on strictly local observables $\hA$ appearing in the definition of asymptotic photon fields (\ref{intro-as-field}),
$B$  is much less sensitive to the background radiation mentioned above. Thus, the tentative 
inclusive collision cross sections of the form
\beqa
\lim_{t\to -\infty}\lan \Psi^{\out}, C_{1,t}\ldots C_{\ell,t}\Psi^{\out}\ran \label{inclusive}
\eeqa
are likely to exist. Although available methods allow to control such limits only in massive theories \cite{DG14}, 
their extension to the case of sharp masses  embedded in continuous spectrum is thinkable. Another strategy
may be to consider limits (\ref{inclusive}) in the framework of algebraic perturbative QFT. As a matter of fact,
(\ref{inclusive}) bears some similarity to expressions studied in the book of Steinmann \cite[formula (16.38)]{St}.

\vspace{0.5cm}

\nin\textbf{Summary.} Our paper is organized as follows. In Section~\ref{framework-section} we discuss
 Haag-Kastler nets and their representations. Section~\ref{Preliminaries}  
surveys various preliminary results, most of which concern the energy-momentum 
transfer of observables. In Section~\ref{asymptotic-fields-section} we introduce the asymptotic photon
fields approximants. We collect their representation-independent properties, e.g. the uniform 
energy bounds~(\ref{intro-energy-bounds}) and the decomposition into creation/annihilation operators.
In that section we also give the proof of relation~(\ref{annihilation}) which is our main technical result.
In Section~\ref{vacuum} we revisit scattering theory of photons in a vacuum representation. 
In Section~\ref{charged} we combine information from all the earlier sections to construct the outgoing
wave operator of Compton scattering in a hypercone localized representation and show that it is isometric and Poincar\'e covariant.
In Section~\ref{last-section} we construct the Haag-Kastler net of asymptotic photon fields in a hypercone localized representation 
and show that any single-electron state induces a vacuum representation of this net. More technical aspects of our discussion
are postponed to the appendices.   

\vspace{0.5cm}

\nin\textbf{Acknowledgements.} W.D. would like to thank D. Buchholz, M. Duell, C. G\'erard, A. Pizzo and Y. Tanimoto
for interesting discussions on topics related to this paper. Financial support from the Emmy Noether Programme of the
DFG, within the grant DY 107/2-1, is gratefully acknowledged.

\section{Framework}\setcounter{equation}{0} \label{framework-section}
\subsection{Haag-Kastler nets}

Let $M=\real^4$ be the  Minkowski spacetime.
We denote by $\mcK$ the family
of double cones $\mco\subset M$ ordered by inclusion and write $\mco_\cc$ for the causal complement of $\mco$ in 
$M$\footnote{Note the distinction between the causal complements in $M$ and $V$, which is indicated by lower respectively upper indices.}. Furthermore, let $\wt\mcP_+^\uparrow=\real^4\rtimes SL(2,\complex)$ denote the covering group of the proper ortochronous Poincar\'e group 
$\mcP_+^\uparrow$. Its elements $\la=(x,\wt\La)$
act on $M$ via $\la y=\La y+x$, where $\La\in \mcL_+^\uparrow$ is the Lorentz transformation corresponding to $\wt\La\in SL(2,\complex)$.
\bed\label{HK}  We say that $\mcK\ni\mco\mapsto \mfb(\mco)\subset B(\hil)$ is a Haag-Kastler net of von Neumann algebras if
the following properties hold:
\begin{enumerate}
\item[(a)] (Isotony) $\mfa(\mco_1)\subset \mfa(\mco_2)$ for $\mco_{1}\subset \mco_2$.
\item[(b)] (Locality) $\, [\mfb(\mco_{1}), \mfb(\mco_{2})]=0$  for $\mco_{1}\subset \mco_{2, \cc}$.
\item[(c)] (Covariance) There is a continuous unitary representation $\wtU$ of  $\wt\mcP_+^\uparrow$  such that
\beqa
\wtU(\wtla)\mfb(\mco)\wtU(\wtla)^*=\mfb(\la \mco)\quad  \mathrm{for}\quad 
\wtla\in \wt\mcP_+^\uparrow.
\eeqa
\item[(d)] (Positivity of energy) The joint spectrum of the  generators of translations, denoted $\Sp\,(\wtU\res\real^4)$, is contained in the closed future lightcone $\ov{V}_+$.

\end{enumerate}
A Haag-Kastler net will be denoted by $(\mfa,U)$.

\eed
\bed We say that a Haag-Kastler  net describes Wigner particles of mass $m\geq 0$ if there is a subspace $\mfh \subset\hil$ on which 
$\wtU(\la), \la\in \wt\mcP_+^\uparrow$, acts like a representation of mass $m$.
\eed
\nin Further useful definitions are as follows. For any region $\mathcal U\subset M$ we set
\beqa
\mfb_{\loc}(\mathcal U):=\bigcup_{\mco\subset \mathcal U} \mfb(\mco) 
\quad \textrm{ and } \quad \mfb(\mathcal U):=\ov{\mfb_{\loc}(\mathcal U)}^{\|\,\cdot\,\|}.
\eeqa
In particular, we refer to $\mfb_{\loc}:=\mfb_{\loc}(M)$ as the algebra of strictly local operators and to $\mfb:=\mfb(M)$ 
as the global algebra of the net.
For the unitary representation of translations $\wtU\res{\real^4}$ we shall write
$\wtU(x)=\e^{\i(Hx^0-\bP\bx)}$ 
and the joint spectral measure of the energy-momentum operators $(H,\bP)$ shall be denoted by $E(\,\cdot\,)$.
For translated observables $A\in \mfb$ the notations $\al_x(A):=A(x):=\wtU(x)A\wtU(x)^*$ are used. Moreover, we define
\beqa
\mfa_{\loc,0}:=\{A\in \mfa_{\loc}\,|\, x\mapsto A(x)\, \textrm{ smooth in norm}\}.
\eeqa


\subsection{Representations}\label{rep-subsection}
Consider a Haag-Kastler net $(\mfa,\wtU)$ and
let $\pi: \mfa \to B(\hil_{\pi})$ be a (unital) representation. We say that
$\pi$ is (Poincar\'e) covariant, if there exists a strongly continuous unitary representation
$\wtU_{\pi}$ of $\wt\mcP_+^\uparrow$ on $\hil_{\pi}$ such that
\beqa
\wtU_{\pi}(\la)\pi(A)\wtU_{\pi}(\la)^*=\pi(\wtU(\la)A\wtU(\la)^*), \quad A\in \mfa.
\eeqa  
Moreover, we say that $\pi$ has positive energy if  $\Sp\,(\wtU_{\pi}\res\real^4)\subset \ov{V}_+$.
It is easy to see that if $\pi$ is a covariant, positive energy representation, then,
\beqa
\mco\mapsto \mfa_{\pi}(\mco):=\pi(\mfa(\mco))''
\eeqa
is again a Haag-Kastler net which will be denoted $(\mfa_{\pi}, U_{\pi})$. 
\bed\label{vacuum-def} If $\pi$ is an irreducible, covariant, positive energy representation   and $\hil_{\pi}$ contains a unique (up to a phase) unit vector $\Om$, invariant under $\wtU_{\pi}$, then we say that $\pi$ is a vacuum representation.
\eed

In order to proceed to charged representations, we choose an open future lightcone $V$ and denote for any region $\mathcal U\subset V$ by $\mathcal{U}^{\cc}$ its causal complement in $V$. Next,
we define  a class of regions in $V$ which are called \emph{hypercones} in \cite{BR14}. We recall here briefly their definition referring to \cite{BR14} for more details: Choose coordinates so that 
$V=\{\, x\in\real^4\,|\,  x_0>|\bx| \, \}$ and fix a hyperboloid
$\hyp_{\bar\tau}=\{\, x\in V\,|\,  x_0=\sqrt{  \bx^2+\bar\tau^2} \, \}$ for some $\bar\tau>0$. Project $\hyp_{\bar\tau}$ through the origin onto the plane $x_0=1$
so as to identify it with the open unit ball $\ball\subset\real^3$. This projection is the Beltrami-Klein model of hyperbolic geometry. 
Consider the family of (truncated) pointed convex Euclidean cones
$\mathsf K$ in $\ball$ with elliptical bases. It gives rise to a Lorentz invariant family
of hyperbolic cones $\mathsf C=\mathsf C(\mathsf K)$ in $\hyp_{\bar\tau}$. A hypercone $\mcC=\mcC(\mathsf K)$ is the causal completion of such $\mathsf C$, i.e. $\mcC=\mathsf C^{\cc\cc}$, and the family of all hypercones as described above is denoted by $\mcF_{V}$. We recall that for $\mathsf K\cap \mathsf K'=\emptyset$ we have
that $\mcC(\mathsf K)$ and $\mcC(\mathsf K')$ are spacelike separated. 
\bed Let $(\mfa,\wtU)$ be a Haag-Kastler net in a  vacuum representation
and let $\pi$ be a covariant positive energy representation.
We say that $\pi$ is hypercone localized  if for any future lightcone $V$
and  $\mcC\in\mcF_{V}$ there exists a unitary $W_{\mcC}:\hil\to\hil_{\pi}$ such that
\beqa
\pi(A)=W_{\mcC}A W_{\mcC}^*\quad  \mathrm{for}\quad A\in \mfa(\mcC^{\cc}).
\eeqa
\eed
\begin{remark}
 It is easy to see that the morphisms $\si_{\mcC,M}:\mfa\to B(\hil)$ from \cite{BR14} are irreducible, hypercone localized representations. 
\end{remark}

Note that for any hypercone localized representation $\pi$ and any $\mco\in\mcK$ we have
$\pi(\mfa(\mco))=\pi(\mfa(\mco))''$ and, therefore, $\pi(\mfa)=\mfa_{\pi}$. It is also easy to see that
any hypercone localized representation is faithful.

\section{Preliminaries}\label{Preliminaries}\setcounter{equation}{0}
In this section we consider an arbitrary Haag-Kastler net  $(\mfa,U)$.
\subsection{Almost local operators}\label{almost-local-subs}

The following standard class of observables shall be of particular importance in Lemma~\ref{HA-lemma} stated below.
\bed An operator $A\in \mfb$ is called almost local if there exists a family of local observables $A_r$, localized in standard double cones $\mco_r$ of radius $r$ centered at zero, and for any $n\in \nat$ there is a constant $C_n$ such that
\beqa
\|A-A_r\|\leq\frac{C_n}{r^n}.
\eeqa
\eed
\nin Next, we define for any $B\in \mfb$ the  smeared operators 
\beqa
B(g):=\left\{ \begin{array}{ll} 
\int d^3x\, B(\bx)g(\bx) & \textrm{for $g\in S(\real^3)$,} \\
\int d^4x\, B(x)g(x)   & \textrm{for $g\in S(\real^4)$.} \end{array} \right.
\eeqa 
Since local algebras are von Neumann, $B(g)\in \mfb$.
It is easy to see that for $B\in \mfb_{\loc}$ the operators $B(g)$, as defined above, 
are almost local.
\subsection{Arveson spectrum}\label{Arveson-subs}

\bed  For $B\in \mfb$ we define the Arveson spectrum of $B$ as 
the support
of the Fourier transform of $\real^4\ni x\mapsto B(x)$, understood as an operator 
valued distribution. That is,
\beqa
\Sp_{B}\al:=\ov{ \bigcup_{\Psi,\Phi\in \hil} \supp\,\lan \Psi, \wt B(\,\cdot\,)\Phi\ran}.
\eeqa
\eed 
\nin This concept is useful due to the \emph{energy-momentum transfer relation} which gives~\cite{Ar82}
\beqa
BE(\De)\hil\subset E(\ov{\De+\Sp_{B}\al})\hil \label{EM-transfer}
\eeqa
for any Borel set $\De\subset \real^4$. For future reference, we also note the simple fact that
\beqa
\Sp_{B(g)}\al\subset \supp\, \wt g, \quad g\in S(\real^4),
\eeqa
which allows to construct almost local observables whose Arveson spectrum is contained in a prescribed set.
We refer to Appendix~\ref{conventions} for our conventions concerning the Fourier transform.
\subsection{Energy bounds}

In this subsection we recall results from \cite{Bu90} and \cite{He14} which are  important for our analysis. They 
are  proven by combination of the energy-momentum 
relation (\ref{EM-transfer})  and almost locality. 

The following
lemma enters into the proof of Proposition~\ref{a-B-prop} stated below which is then used to obtain our  main technical result, 
namely Theorem~\ref{main-ergodic}.
\bel\emph{\cite{Bu90}} \label{HA-lemma} Let $B\in \mfb$ be almost local and such that $\Sp_B\al$ is a compact
set which does not intersect with $\ov{V}_+$. 
Then, for any compact $\De\subset \real^4$ there
exists a constant $c_{\De}$ such that for any compact $K\subset \real^3$
\beqa
\|E(\De)\int_{K} d^3x\, (B^*B)(\bx) E(\De)\|\leq c_{\De}.
\eeqa
\eel
\nin
Next, we state a result which is at the basis of Proposition~\ref{uniform-bound-prop} stated below, giving  information about the domains of asymptotic fields. 
\begin{proposition}\emph{\cite{Bu90, He14}}\label{herdegen-bound-new}
	Let $A\in\mathfrak{A}_{\mathrm{loc},0}$ and $n\in \real^4$ be a unit future oriented timelike vector, i.e. $n_0=\sqrt{1+\bn^2}$. Then, for any $g\in S(\real^4)$ 
\begin{equation}
	\|A( (n_{\mu}\pa^{\mu})^3 g)(1+H)^{-1}\|\leq c\, \sup_{\ell=0,1} \|\pa_0^{\ell}\widetilde{g}\|_2, \label{energy-bound-estimate}
	\end{equation}
where the constant $c$ is independent of $g$. 

\end{proposition}
\nin First bounds of this type were proven in \cite{Bu90}. The above variant can be inferred from \cite{He14} as follows. Starting
with Theorem 4 of this reference, one can replace operators $\wt A_{\pm}^{k}(p):=\e^{-\i k\pi/2}\theta(\pm p^0)|p^0|^k\wt A(p)$ 
with $\wt A_{\pm}^{k,n}(p):=\e^{-\i k\pi/2}\theta(\pm n_{\mu}p^{\mu})|n_{\mu}p^{\mu}|^k\wt A(p)$, where $k>0$ and $n$ 
is chosen as in 
Proposition \ref{herdegen-bound-new}. Then, formula (20) of  \cite{He14} gives (\ref{energy-bound-estimate}). This coordinate
frame independence was actually noticed and used in the proof of Theorem 5 (ii) of \cite{He14}.

\subsection{Auxiliary maps $a_B$}

In this subsection we recall some concepts and facts from \cite{DG14}. For any $B\in \mfb$ 
we have the continuous map $a_B: \hil\to S'(\real^3;\hil)$ given by
\beqa
(a_B\Psi)(\bx)=B(\bx)\Psi, \quad \bx\in \real^3.
\eeqa
Its dual $a_B^*: S(\real^3;\hil)\to \hil$ is given by
\beqa
a_B^*\Phi=\int d^3x\, B^*(\bx)\Phi(\bx). 
\eeqa
We identify  $S'(\real^3;\hil)=\hil\otimes S'(\real^3)$ and define for
$g\in S(\real^3)$ the 
functionals $(1_{\hil}\otimes \lan \ov g|): S'(\real^3;\hil)\to \hil$ by
\beqa
(1_{\hil}\otimes \lan \ov g|)\Psi=\int d^3x\, g(x)\Psi(x).
\eeqa 
Their adjoints are denoted by $(1_{\hil}\otimes |g\ran)$. 
For future reference, we note the identities 
\begin{align}
B(g)&=(1_{\hil}\otimes \lan \ov g|)\circ a_B, \label{smearing-vs-aB}\\
B^*(g)&=a_B^*\circ (1_{\hil}\otimes |g\ran).
\end{align}
For $B$ as in Lemma~\ref{HA-lemma} the maps $a_B$ have the following useful properties.
\bep\cite{DG14}\label{a-B-prop} Let $B$ be almost local and such that $\Sp_B\al$ is a compact
set which does not intersect with $\ov{V}_+$. Furthermore, let  $\De\subset \real^4$ be compact. 
Then:
\begin{itemize}
\item[(a)]  $a_B E(\De): \hil\to \hil\otimes L^2(\real^3)$ is bounded.
\item[(b)] $ a_B E(\De)\circ f(\bP)=f(\bP+D_{\bx})\circ a_B E(\De)$
for any $f\in L^{\infty}(\real^3)$.
\end{itemize}
Here we set $D_{\bx}=-\i\nabla_{\bx}$ and  use the shorthand notation
$\bP+D_{\bx}$ for $\bP\otimes 1_{L^2(\real^3)}+1_{\hil}\otimes D_{\bx}$.
\eep
\proof (a) follows from  the identity
\beqa
E(\De)a_B^*\circ a_B E(\De)=E(\De)\int d^3x\, (B^*B)(\mathbf{x})E(\De)
\eeqa
and Lemma~\ref{HA-lemma}. To verify~(b), one first checks that 
\beqa
a_B E(\De)\circ \e^{-\i\by\bP}= \e^{-\i\by(\bP+D_{\bx})} \circ a_B E(\De).
\eeqa
 Then,  the fact follows from properties of the Fourier transform and approximating arguments.\hfill \qed\\
In view of Proposition~\ref{a-B-prop},  for any
compact $\De$  the identity
\beqa
B(g)E(\De)=(1_{\hil}\otimes \lan \ov g|)\circ a_BE(\De)
\eeqa
extends by continuity to $g\in L^2(\real^3)$.
\subsection{Action of Lorentz transformations on a sphere}\label{Lorentz-subsection}
We conclude this preliminary section with a brief consideration
about Lorentz transformations which will be used to show Poincar\'e
covariance of our constructions. 

Given $\La\in \mcL_+^{\uparrow}$, we divide the expression $(\La x)^{\mu}=\La^{\mu}_{\ph{0}\nu} x^{\nu}$
into its time and space parts, i.e.
\beqa
& &(\La x)^{0}=\lan \bv_{\La} \ran x^0+  \bv_{\La}\bx,\\
& &(\La x)^i=- \bv_{\La^{-1}}^i x^0+[\La]^i_{\ph{0}j} \bx^j,
\eeqa
where we set $\bv_{\La}^i:=\La^{0}_{\ph{0}i}$, $[\La]^i_{\ph{0}j}:=\La^i_{\ph{0} j}$, $i,j=1,2,3$. Note that 
$\bv_{\La^{-1}}^i=-\La^i_{\ph{0}0}$ and $\La^{0}_{\ph{0}0}
=\sqrt{1+|\bv_{\La}|^2}=:\lan \bv_{\La} \ran$. 
Now let $S^2=\{\mathbf{n}\in\mathbb{R}^3\,|\, |\mathbf{n}|=1\}$ be the unit sphere. 
For $\bn\in S^2$ we get
\beqa
& &\La(1,\bn)=(\lan \bv_{\La} \ran +  \bv_{\La} \bn)(1, g_{\La}(\bn)), \label{g-Lambda-zero}\\
& &g_{\La}(\bn):=\fr{-\bv_{\La^{-1}} +[\La] \bn }
{|- \bv_{\La^{-1}}+[\La] \bn |}, \label{g-Lambda}
\eeqa
since $\La(1,\bn)$ is a lightlike vector.
Here $g_{\La}: S^2\to S^2$ is a family of diffeomorphisms of $S^2$, which in
fact forms a representation of $\mcL_+^{\uparrow}$, i.e. $g_{\La_1\La_2}=g_{\La_1}\circ g_{\La_2}$. It is, moreover, continuous in the
following sense
\beqa
\lim_{\La\to I}\|g_{\La}-g_{I}\|_{\infty}=0. \label{Lorentz-continuity}
\eeqa

Finally, we recall that by using the multiplication rules of the Poincar\'e group, the spectrum condition and the Stone theorem one obtains invariance of the domain of any positive power of the Hamiltonian $D(H^n)$ under the action of $U(\la), \la\in \wt\mcP_+^{\uparrow}$.  
Setting $P:=(H, \bP)$, the standard relation follows
\beqa
U(\wt \La)P^{\mu} U(\wt \La)^*=(\La^{-1})^{\mu}_{\ph{0}\nu} P^{\nu},
\eeqa
in the sense of operators on $D(H^n)$. 
With the above definitions, we obtain on $D(H^n)$
\beqa
& &U(\wt \La)HU(\wt \La)^* =\lan \bv_{\La^{-1}} \ran H + 
\bv_{\La^{-1}}\bP,\label{H-transform}\\
& &U(\wt \La)\bP U(\wt \La)^* =- \bv_{\La} H+[\La^{-1}] \bP. \label{P-transform}
\eeqa
Denoting by $P_{\pho}$ the projection on the single-photon subspace
$\mfh_{\pho}$, it further follows for any $f\in C^{\infty}(S^2)$ that
\beqa
U(\wt \La)f\left(\tfrac{\bP}{|\bP|}\right) U(\wt \La)^*P_{\pho}=
(f\circ g_{\La^{-1}}) \left(\tfrac{\bP}{|\bP|}\right) P_{\pho}. \label{transform}
\eeqa


\section{Asymptotic photon fields}\setcounter{equation}{0} \label{asymptotic-fields-section}

In this section the pair $(\mfa,U)$ still refers to an arbitrary Haag-Kastler net.

\subsection{Spherical means}
For further purposes we introduce the following Poincar\'e invariant subset of $S(\real^4)$
\beq
S_*(\real^4):=\{ (n_{\mu}\pa^{\mu})^{5}g \,|\, g\in S(\real^4), \,\, n_0=\sqrt{1+\bn^2}\,\}, \label{S-star-def}
\eeq
as well as certain Poincar\'e invariant subsets of $\mfa$, namely
\beqa
& &\mfb_{S_*}:=\{\, B(g)  \,|\,  B\in \mfa_{\loc,0}, \
g\in S_*(\real^4) \,\}, \\
& &\mfb^{S_*}:=\Span\, \mfb_{S_*}, 
\eeqa
where $\Span$ denotes finite linear combinations.
For any $A\in \mfb^{S_*}$, $f\in C^{\infty}(S^2)$,  we set   as in \cite{Bu82}
\beqa
A_t\{f\}:=-2\,t\int d\om(\bn)\,f(\mathbf{n})\,\partial_0A(t,t\mathbf{n}).
\eeqa
Here $d\om(\mathbf{n})=\frac{\sin\nu\,d\nu d\varphi}{4\pi}$ is the normalized, invariant measure on $S^2$ and 
$\pa_0A:=\partial_s(\e^{\i sH}A\e^{-\i sH})|_{s=0}$. In order to improve the convergence in the limit of large $t$, we proceed to  time averages of $A_t\{f\}$, namely
\begin{equation}\label{timeAverage}
\bar{A}_t\{f\}:=\int\,dt'\,h_t(t')\,A_{t'}\{f\}, 
\end{equation}
where for non-negative $h\in C_0^{\infty}(\real)$, supported in the interval $[-1,1]$ and normalized so that
$\int dt\, h(t)=1$, we set $h_t(t')=t^{-{\beps}} h(t^{-{\beps}}(t'-t))$ with
$t\geq 1$ and $0<\beps<1$. 

For the discussion of asymptotic creation and annihilation operators in Subsection~\ref{creation-annihilation-subs} we need $(n_{\mu}\pa^{\mu})^{5}$
in (\ref{S-star-def}), rather than just $(n_{\mu}\pa^{\mu})^{3}$ from the energy bound~(\ref{energy-bound-estimate}). For the same purpose,
it is important to use Schwartz class functions in (\ref{S-star-def}). Since strict locality plays a crucial role in the later part of this paper,
we also define the following sets
\begin{align}
C_*(\real^4)&:=\{ (n_{\mu}\pa^{\mu})^{5}g \,|\, g\in C_0^{\infty}(\real^4), \,\, n_0=\sqrt{1+\bn^2}\,\}\subset S_*(\real^4), \\
\mfb_{C_*}&:=\{\, B(g)  \,|\,  B\in \mfa_{\loc,0}, \ g\in C_*(\real^4) \,\}\subset \mfb_{S_*}\cap \mfa_{\loc,0}, \\
\mfb^{C_*}&:=\Span\, \mfb_{C_*}, \\
\mfb_{C_*}(\mco)&:=\mfb_{C_*}\cap \mfb(\mco), \quad \mfb^{C_*}(\mco):=\mfb^{C_*}\cap \mfb(\mco), 
\quad  \mco\in \mcK. \label{local-C}
\end{align}
The linear structure of $\mfb^{S_*}$ and $\mfb^{C_*}$ will be important  in Section~\ref{last-section}.

\subsection{Fourier space representation}
Given $A\in \mfb_{S_*}$, a convenient representation for $A_t\{f\}$ can be found, which will be frequently used
in the remaining part of this section and in the proof of Lemma~\ref{StrongVacuum} below. 
 This representation is stated in the following lemma.
\begin{lemma}\label{spherical-means-analysis} 	Let $A\in \mfb_{S_*}$, i.e. $A=B(g)$, where $B\in \mfa_{\loc,0}$ and 
$g\in S_*(\real^4)$.  Then,
	\beqa\label{first-equality}
	A_t\{f\}=(\pa_0 B)(g\ast_3 f_t)(t), \quad \mathrm{ where } \quad f_t(\mathbf{x}):=-\fr{1}{4\pi}\frac{2}{|\mathbf{x}|}\,\delta(t-|\mathbf{x}|)\,f\left(\frac{\mathbf{x}}{|\mathbf{x}|}\right),
	\eeqa
where $\ast_3$ is defined in Appendix~\ref{conventions}.
	Moreover, the Fourier transform of $g\ast_3 f_t\in S(\real^4)$  has the following form
	\begin{align}
	(\wt{g\ast_3 f_t})(p)=&\, \fr{\wt{g}(p)}{\i|\mathbf{p}|}     
	\left(f\left( \fr{\bp}{|\bp|}   \right)  \e^{-\i t|\mathbf{p}|}-    f\left(-\fr{\bp}{|\bp|}  \right)  \e^{\i t|\mathbf{p}|}
	+\int_{0}^{\pi} d\nu\,F(\bp,\nu)   \e^{-\i t|\mathbf{p}|\cos \nu}\right), \label{convolution-formula}
	\end{align}
	where $F$ is a bounded measurable function depending on $f$. (In particular, $F=0$ if $f=\mathrm{const}.$)
\end{lemma}
\begin{proof} The equality $A_t\{f\}=\pa_0B(g\ast_3 f_t)(t)$, with $f_t$ given by (\ref{first-equality}), is straightforward to check.
Since $(\wt{g\ast_3 f_t})(p)=(2\pi)^{3/2}\wt g(p) \wt f_t(\bp)$, it remains to compute
\begin{align}
\widetilde{f_t}(\mathbf{p})&=(2\pi)^{-3/2}\int d^3x\, \e^{-\i \mathbf{p}\mathbf{x}} f_t(\mathbf{x})\non\\
&=\frac{-2}{(2\pi)^{3/2}}\int d^3x\, \e^{-\i\mathbf{p}\mathbf{x} } \fr{1}{4\pi} \frac{1}{|\mathbf{x}|}\,\delta(t-|\mathbf{x}|)\,f\left(\frac{\mathbf{x}}{|\mathbf{x}|}\right)\non\\
&=\frac{-2t}{(2\pi)^{3/2}}   \int d\omega(\mathbf{n})\,
\e^{-\i t\mathbf{n}\mathbf{p}} f(\mathbf{n}).
\end{align}
A coordinate independent treatment of Fourier transforms on the sphere can be found in \cite{DH15}. We give here
an elementary coordinate dependent computation. To this end, we choose a measurable family of rotations $\bp\mapsto R_{\bp}\in SO(3)$ such that $R_{\bp} \textbf{e}_3=\mathbf{p}/|\mathbf{p}|$,
where $\textbf{e}_3=(0,0,1)$ in the reference system in which $d\om(\mathbf{n})=\frac{\sin\nu\,d\nu d\varphi}{4\pi}$. We obtain
\begin{align}
\widetilde{f_t}(\mathbf{p})=&\,\frac{-t}{(2\pi)^{5/2}}\int_{0}^{2\pi}d\varphi\,\int_{0}^{\pi} d\nu\,\sin\nu\, \e^{-\i t|\mathbf{p}|\cos \nu} f(R_{\bp}\mathbf{n}(\varphi,\nu))\non\\
=&\,\frac{-1}{(2\pi)^{5/2}\i|\mathbf{p}|}\int_{0}^{2\pi}d\varphi\,\int_{0}^{\pi} d\nu\,f(R_{\bp}\mathbf{n}(\varphi,\nu)) \pa_{\nu} \e^{-\i t|\mathbf{p}|\cos \nu}.
\end{align}
Finally, integrating by parts and noting that $\bn(\varphi,0)=\textbf{e}_3, \bn(\varphi,\pi)=-\textbf{e}_3$, we arrive~at
\begin{align}
\widetilde{f_t}(\mathbf{p})=&\frac{-1}{(2\pi)^{5/2}\i|\mathbf{p}|}\int_{0}^{2\pi} d\varphi\,  
\big(f(R_{\bp}\mathbf{n}(\varphi,\pi))  \e^{\i t|\mathbf{p}|}-f(R_{\bp}\mathbf{n}(\varphi,0))  \e^{-\i t|\mathbf{p}|}\big)\non\\
&+\frac{1}{(2\pi)^{5/2}\i|\mathbf{p}|}\int_{0}^{2\pi}d\varphi\,\int_{0}^{\pi} d\nu\,\pa_{\nu}f(R_{\bp}\mathbf{n}(\varphi,\nu))  \e^{-\i t|\mathbf{p}|\cos \nu}\non\\
=&\frac{-(2\pi)}{(2\pi)^{5/2}\i|\mathbf{p}|} 
\left(f\left(-\fr{\bp}{|\bp|}\right)  \e^{\i t|\mathbf{p}|}-f\left( \fr{\bp}{|\bp|}\right) \e^{-\i t|\mathbf{p}|}\right)\non\\
&+\frac{1}{(2\pi)^{5/2}\i|\mathbf{p}|}\int_{0}^{2\pi}d\varphi\,\int_{0}^{\pi} d\nu\,\pa_{\nu}f(R_{\bp}\mathbf{n}(\varphi,\nu))  \e^{-\i t|\mathbf{p}|\cos \nu}.
\end{align}
This completes the proof.\hfill \qed \end{proof}

\subsection{Uniform energy bounds} \label{uniform-bounds}

This subsection is concerned with uniform bounds on $t\mapsto \bar{A}_t\{f\}$. The following result holds.
\begin{proposition}\label{uniform-bound-prop}
	Let $A\in \mfb_{S_*}$, i.e. $A=B((n_{\mu}\pa^{\mu} )^{5}g')$, $B\in\mfa_{\loc,0}$ and $g'\in S(\real^4)$. Then,
	\begin{equation}\label{bound1}
	\sup\limits_{t\in\left[1,\infty)\right.}\|\bar{A}_t\{f\}(1+H)^{-1}\|
\leq c\sup_{\ell=0,1}\||\mathbf{p}|^{-1}  \pa_0^{\ell}  \big((n_{\mu}p^{\mu})^2\wt{g'})\|_2<\infty.
	\end{equation}
The constant  $c$ above  is independent of $g'$.
\end{proposition} 
\begin{proof}
	By Lemma \ref{spherical-means-analysis} we have that $A_t\{f\}=(\pa_0 B)(g \ast_3 f_t)(t)$, where $g=(n_{\mu}\pa^{\mu})^{5}g'$ and
$g'\in S(\real^4)$.  Therefore, $g \ast_3 f_t=(n_{\mu}\pa^{\mu})^{5} (g' \ast_3 f_t)$, with $g' \ast_3 f_t\in S(\real^4)$.
Thus, Proposition~\ref{herdegen-bound-new} yields
\beqa
\|\bar{A}_t\{f\}(1+H)^{-1}\|\leq c\sup_{\ell=0,1}\|\pa_0^{\ell}\big((n_{\mu} p^{\mu})^2(\widetilde{g' \ast_3 f_t})\big)\|_2,
\eeqa
where $c$ is independent of $t$. Now by formula~(\ref{convolution-formula}) we have
\begin{equation}
\left| \pa_0^{\ell}\big((n_{\mu} p^{\mu})^2(\widetilde{g' \ast_3 f_t})(p) \right|\leq c'\, |\mathbf{p}|^{-1}  |\pa_0^{\ell}  \big((n_{\mu}p^{\mu})^2\wt{g'}(p)\big)|,
\end{equation}
where $c'$ is independent of $p$, $t$ and $g'$.  This completes the proof. 
\hfill\qed   \end{proof}
In the following we shall be interested in the convergence of $\bar{A}_t\{f\}$, $A\in \mfa^{S_*}$, 
to a limit $A^{\out}\{f\}$ as $t\to\infty$. To start with,  we define $A^{\out}\{f\}$ as an operator on the domain 
\beqa
D_{\max}(A,f):=\{ \Psi\in \hil\,|\, A^{\out}\{f\}\Psi:=\lim_{t\to\infty}\bar{A}_t\{f\}\Psi \textrm{ exists} \}.
\eeqa
(For $f\equiv 1$ we will abbreviate $D_{\max}(A,f)$ by $D_{\max}(A)$).
Note that $D_{\max}(A,f)$ may depend on $A$, $f$, may not be Poincar\'e invariant  
and a priori may even be trivial.  Another domain we shall be interested in is
\beqa
D_H:=\bigcap_{n\geq 1} D(H^n),\label{D-H-definition}
\eeqa
where $D(H^n)$ is the domain of self-adjointness of the $n$-th power of the Hamiltonian~$H$. It is easy to see that $D_H$ is
dense and Poincar\'e invariant. The next result can be inferred from Proposition~\ref{uniform-bound-prop}  and the discussion in Appendix~\ref{admissible-appendix}.
\bep\label{general-admissible-two} 
Let $i=1,\ldots, n$ and suppose that the domains $D_{\max}(A_i,f_i)$ and $D_{\max}(A_i^*,\bar f_i)$  are dense. Then, we have
\begin{enumerate}
\item[(a)] $D_H\subset D_{\max}(A_i,f_i), D_H\subset D_{\max}(A_i^*,\bar f_i)$,
\item[(b)]  $A_i^{\out}\{f_i\} D_H\subset D_H$,
\item[(c)] $A_1^{\out}\{f_1\}\ldots A_n^{\out}\{f_n\}\Psi=\lim_{t\to\infty}\bar{A}_{1,t}\{f_1\}\ldots \bar{A}_{n,t}\{f_n\}\Psi$ for $\Psi\in D_H$.
\end{enumerate}
The operators $A_i^{\out}\{f_i\}\res D_H$ are closable and uniquely determined by the values of $A_i^{\out}\{f_i\}$ on any dense subspace of
 $D_{\max}(A_i,f_i)$.
\eep
\subsection{Asymptotic creation/annihilation operators}\label{creation-annihilation-subs}

Another consequence of the uniform bounds is the existence of asymptotic creation
and annihilation operators under the assumptions of Proposition~\ref{general-admissible-two}. In fact  a similar observation was made in \cite{DH15}.

To construct these operators we proceed as follows. Let $\theta\in C^{\infty}(\real)$, $0\leq \theta\leq 1$, be supported in $(0,\infty)$ and equal to one on $(1,\infty)$. Moreover, let $\be\in C_0^{\infty}(\real^4)$, $0\leq \be\leq 1$, be equal to one in some neighbourhood of zero and satisfy $\be(-p)=\be(p)$. Furthermore, for a parameter $1\leq r<\infty$ and a future oriented timelike unit vector $n$ we define
\beqa
\wt\eta_{\pm,r}(p):=\theta(\pm r (n_{\mu}p^{\mu}))\be(r^{-1} p). \label{test-functions-zero}
\eeqa 
As $r\to \infty$ these functions approximate the characteristic functions of the positive/negative energy half planes $\{\, p\in \real^4\,|\,\pm n_{\mu}p^{\mu}\geq 0\,\}$. We also have $\bar \eta_{\pm,r}=\eta_{\mp,r}$. Note that the family of functions $\eta_{\pm,r}$, as specified above, is invariant under Lorentz transformations.
\bep\label{creation-annihilation} Let $A\in \mfb_{S_*}$, $f\in C^{\infty}(S^2)$. Suppose that $D_{\max}(A,f), D_{\max}(A^*,\bar f)$ are dense and
the timelike unit vectors $n$ entering the definition of $A$ and of $\eta_{\pm,r}$ coincide. Then:
\begin{enumerate}
\item[(a)] The limits 
$A^{\out}\{f\}^{\pm}\Psi:=\lim_{r\to \infty}A^{\out}\{f\}(\eta_{\pm,r})\Psi$, $\Psi\in D_H$, 
exist and define the creation and annihilation parts of $A^{\out}\{f\}$ as  operators on $D_H$. $A^{\out}\{f\}^{\pm}$
do not depend on the choice of functions $\theta$ and $\be$ in (\ref{test-functions-zero}) within the specified restrictions.

\item[(b)]  $(A^{\out}\{f\}^{\pm})^*\res D_{H}=A^{*\out}\{\bar f\}^{\mp}$. In particular,  $A^{\out}\{f\}^{\pm}$ are closable operators.

\item[(c)] $A^{\out}\{f\}^{\pm}D_H\subset D_H$.

\item[(d)] $A^{\out}\{f\}=A^{\out}\{f\}^{+}+A^{\out}\{f\}^{-}$ on $D_H$.
\end{enumerate}

\eep
\begin{remark}\label{Span} The proposition can be generalized to $A\in \mfb^{S_*}$ as follows. Consider a decomposition $A=\sum_{i=1}^{\ell}A_i$, $A_i\in \mfb_{S_*}$, and 
assume that $D_{\max}(A_i,f)$ and $D_{\max}(A_i^*,\bar f)$ are dense. Define
$A^{\out}\{f\}^{\pm}:=\sum_{i=1}^{\ell}A_i^{\out}\{f\}^{\pm}$
on $D_H$. Then it is easy to see that $A^{\out}\{f\}^{\pm}$ satisfy the properties
(b),(c) and (d) of the proposition.
\end{remark}
\proof  (a) Making use of Propositions~\ref{uniform-bound-prop} and \ref{general-admissible-two}, we compute
for  $1\leq r_1\leq r_2$ and $\Psi\in D_H$, that
\begin{align}
\|A^{\out}\{f\}(\eta_{\pm, r_1}-\eta_{\pm,r_2})    \Psi\|
&=\lim_{t\to\infty}\| A_t\{f\}(\eta_{\pm,r_1}-\eta_{\pm,r_2})\Psi\|\non\\
&\leq c\sup_{\ell=0,1}\||\mathbf{p}|^{-1}  \pa_0^{\ell}  \big((n_{\mu}p^{\mu})^2( \wt\eta_{\pm,r_1}-\wt\eta_{\pm,r_2} )(p)\wt{g'})\|_2\non\\
&\leq c \sup_{\ell=0,1}\int_{r_1}^{r_2} dr\||\mathbf{p}|^{-1}  \pa_0^{\ell}  \big((n_{\mu}p^{\mu})^2  (\pa_r\wt\eta_{\pm,r})(p)\wt{g'})\|_2,
\label{computation-creation-annihilation}
\end{align}
where $\wt g'\in S(\real^4)$ is defined as in Proposition~\ref{uniform-bound-prop} and the functions of $p$ appearing in 
(\ref{computation-creation-annihilation}) are to be understood as multiplication operators acting on  $\wt{g'}$. 
Using the fact that $\pa \theta$ is compactly supported, and therefore
$|n_{\mu}p^{\mu}|\leq c r^{-1}$ when multiplied by $\pa\theta(\pm r (n_{\mu}p^{\mu}))$, it is easy to check that
\beqa
\big|\pa_0^{\ell}\big((n_{\mu}p^{\mu})^2\pa_r\wt\eta_{\pm,r}(p)\big)\big|\leq \fr{c}{r^2}(1+|p|^3), \quad \ell=0,1, \label{formula-to-be-checked}
\eeqa
for $c$ independent of $p$ and $r$. This completes the proof of convergence. Independence of the choice of the functions $\theta$ and $\beta$ is shown
by a similar computation.

(b) We note that for $\Phi, \Psi\in D_H$
\begin{align}
\lan \Phi, A^{\out}\{f\}^{\pm}\Psi\ran&=\lim_{r\to\infty} \lan\Phi,  A^{\out}\{f\}(\eta_{\pm,r})\Psi\ran\non\\
&=\lim_{r\to\infty} \lan A^{*\out}\{\bar f\}(\eta_{\mp,r})\Phi,  \Psi\ran=\lan A^{*\out}\{\bar f\}^{\mp}\Phi, \Psi\ran.
\end{align}

(c) It suffices to set $A_t:=\bar{A}_{t}\{f\}(\eta_{\pm,r})$ in formula~(\ref{commutator-formula}) and take first
the limit $t\to\infty$ and then $r\to\infty$.

(d) We choose a function $\ga\in C_0^{\infty}(\real)$, $0\leq \ga\leq 1$,
such that
\beqa
\theta(-k)+\ga(k)+\theta(k)=1, \quad k\in \real,
\eeqa
and set $\wt\eta_r(p):=\ga_r(r (n_{\mu}p^{\mu}) )\be(r^{-1}p)$. Since $\ga$ is compactly supported,  we have for $\Psi\in D_H$
\beqa
\|A^{\out}\{f\}(\eta_r)\Psi\|\leq c\sup_{\ell=0,1}\||\mathbf{p}|^{-1}  \pa_0^{\ell}  \big((n_{\mu}p^{\mu})^2 \wt\eta_{r}(p)\wt{g'})\|_2\leq 
c'r^{-1}.
\eeqa
Hence, $\lim_{r\to\infty} A^{\out}\{f\}(\eta_r)\Psi=0$, which completes the proof. \hfill\qed

\subsection{Asymptotic vacuum structure}

In this subsection we state and prove our main technical result which is Theorem~\ref{main-ergodic} below. This theorem
will be useful in the proof of the tensor product structure of scattering states in the
vacuum representation (part (a)) and in charged representations (part (b)).
\bet\label{main-ergodic} Let $\eta\in S(\real^4)$ be such that $\wt \eta$ is  supported outside of $\ov{V}_+$. Let $A\in \mfb_{S_*}$ and
$f\in C^{\infty}(S^2)$. Then, for $\Psi\in  E(\mathsf{H}_m)\hil\cap D_H$ and $\mathsf{H}_m=\{\, p\in \real^4\,|\, p^0=\sqrt{\bp^2+m^2}\,\}$, we have:
\begin{enumerate}
 \item[(a)] For $m=0$, $\lim_{t\to\infty}(1-E(\{0\}))  \bar{A}_t\{f\}(\eta)\Psi=0$.
\item[(b)] For $m>0$,  $\lim_{t\to\infty}\bar{A}_t\{f\}(\eta)\Psi=0$.
\end{enumerate}
\eet
\proof 
To begin with, we assume  that $\wt\eta$ is compactly supported and $\Psi=E(\De)\Psi$ for some compact $\De$. Making use of Lemma~\ref{spherical-means-analysis}, we have 
$A_t\{f\}(\eta)=(\pa_0B(\eta)) ( g\ast_3 f_t)(t)$, where
\begin{equation}
(g\ast_3 f_t)(x)=(2\pi)^{-2}\int_0^{\pi} d\mu(\nu)\int d^4p\, \,
\wt{f}_{\nu}(p)\,\e^{-\i p x}\e^{-\i\cos\,\nu |\bp| t}.
\end{equation}
Here $d\mu(\nu):=d\nu+\de(\nu)d\nu+\de(\nu-\pi)d\nu$,
$(p,\nu)\mapsto \wt{f}_{\nu}(p)$ is absolutely integrable, smooth in $p^0$ and
\beqa
\sup_{\nu\in[0,\pi]}(\|\wt{f}_{\nu}\|_2+\|\pa_0 \wt{f}_{\nu}\|_2)<\infty. \label{dominated-convergence-bound}
\eeqa
We set $B':=\pa_0B(\eta)$ and note that it is almost local and $\Sp_{B'}\al$ is a compact set outside of $\ov{V}_+$. Setting $\wt f_{\nu}^t(p):=\wt f_{\nu}(p)\e^{-\i  \cos\,\nu |\bp| t  }$, we have
\begin{align}
\bar{A}_t\{f\}(\eta)\Psi
=&\int dt'\, h_t(t') \e^{\i t' H}B'(g \ast_3 f_{t'}) \e^{-\i t' \om_m(\bP)}\Psi\non\\
=&\int_0^{\pi} d\mu(\nu)\int dt'\, h_t(t') \e^{\i t' H}B'(f_{\nu}^{t'} ) \e^{-\i t' \om_m(\bP)}\Psi.
\end{align}
Now we put $B'_{x^0}(\bx):=B'(x^0,\bx)$, $f_{\nu,x^0}^{t}(\bx):=f_{\nu}^{t}(x^0,\bx)$ and $f_{\nu,x^0}(\bx):=f_{\nu}(x^0,\bx) $. Making use of (\ref{smearing-vs-aB}) and Proposition~\ref{a-B-prop}, we obtain
\beqa
& &\bar{A}_t\{f\}(\eta)\Psi\non\\
& &=\int dx^0\int_0^{\pi} d\mu(\nu)\int dt'\, h_t(t') \e^{\i t' H} (1_{\hil}\otimes 
\lan \ov f^{t'}_{\nu,x^0} |)\circ a_{B'_{x^0}}   \e^{-\i t' \om_m(\bP)}\Psi\non\\
& &=\int dx^0\int_0^{\pi} d\mu(\nu)(1_{\hil}\otimes \lan \ov f_{\nu,x^0}|) \circ\int dt'\, h_t(t') \e^{ \i t'(H- \cos\nu |D_{\bx}|-\om_m(\bP+D_{\bx}) )   }    \circ a_{B'_{x^0}} \Psi.  \non\\
\eeqa
By means of the Dominated Convergence Theorem, the bound~(\ref{dominated-convergence-bound}) and the Mean Ergodic Theorem (Theorem~\ref{MWLem}) we obtain
\beqa
\lim_{t\to\infty}\bar{A}_t\{f\}(\eta)\Psi=\int dx^0\int_0^{\pi} d\mu(\nu)(1_{\hil}\otimes \lan \ov f_{\nu,x^0}|) \circ F_S(\{0\})\circ a_{B'_{x^0}} \Psi,
\eeqa
where $F_S$ is the spectral measure of the operator $S:=H- \cos\nu |D_{\bx}|-\om_m(\bP+D_{\bx})$ on $L^2(\real^3;\hil)$.
To determine $F_S(\{0\})$, we diagonalize $D_{\bx}$ with the help of the Fourier transform. 
We further note that $\|S\Phi\|^2=0$, for some $\Phi=\{\Phi_{\xi}\}_{\xi\in \real^3}\in L^2(\real^3;\hil)$, implies that  $S_{\xi}\Phi_{\xi}=0$ for almost all $\xi$
w.r.t. the Lebesgue measure\footnote{See \cite[Section IV.7]{Ta} for definition and basic properties of $L^2(\real;\hil)$ for non-separable $\hil$.}, where
\beqa
S_{\bxi}:=H+|\bxi|- \om(\bP+\bxi).
\eeqa 
Suppose now that $m=0$.  Then, Proposition~\ref{Ergodic-corollary} gives that $\Phi_{\xi}\in \Ran\, E(\{0\})$
 for $\bla=0$  or $\nu=\pi$  and $\Phi_{\xi}=0$ otherwise. Since $\bla=0$ is of zero Lebesgue
measure, only $\nu=\pi$ contributes and we obtain
 \beqa\label{dctb}
\lim_{t\to\infty}\bar{A}_t\{f\}(\eta)\Psi=\int dx^0(E(\{0\})\otimes \lan \ov f_{\pi,x^0}|)\circ a_{B'_{x^0}} \Psi=E(\{0\})B'(f_{\pi})\Psi.
\eeqa
For  $m>0$ a similar and simpler reasoning gives that the above limit is zero.

It remains to relax the additional assumptions made at the beginning
of the proof. Let, therefore, $\eta$ and $\Psi$ be specified as in the theorem. Then,
by spectral calculus, Proposition~\ref{uniform-bound-prop} and the fact that $\Psi\in D_H$, we have
\beqa
\bar{A}_t\{f\}(\eta)\Psi=E(\De_R)\bar{A}_t\{f\}(\eta)E(\De_R)\Psi+O(R^{-N}), \label{relaxing-assumptions}
\eeqa
where $\De_R:=\{\, p\in \ov{V}_+\,|\, p^0\leq R\,\}$ and $O(R^{-N})$ denotes
a term whose norm is bounded by $C_N/R^N$, with $C_N$ independent of $t$. Making now use of the energy-momentum transfer relation~(\ref{EM-transfer}), we can replace $\eta$ in (\ref{relaxing-assumptions})
by $\eta'$ such that $\wt\eta'$ is compactly supported outside of the future lightcone. Thus, we have
\beqa
\bar{A}_t\{f\}(\eta)\Psi=E(\De_R)\bar{A}_t\{f\}(\eta')E(\De_R)\Psi+O(R^{-N}).\label{relaxing-assumptions-one}
\eeqa
By means of formula~(\ref{relaxing-assumptions-one})
we conclude the proof.\hfill\qed
\bec\label{main-ergodic-corollary} Let $A\in \mfb_{S_*}$ and $f\in C^{\infty}(S^2)$. Suppose further that $D_{\max}(A,f)$ 
and $D_{\max}(A^*,\bar f)$ are dense.
Then, for $\Psi\in  E(\mathsf{H}_m)\hil\cap D_H$, we have:
\begin{enumerate}
 \item[(a)] For $m=0$, $ (1-E(\{0\}))  A^{\out}\{f\}^-\Psi=0$.
\item[(b)] For $m>0$,  $A^{\out}\{f\}^-\Psi=0$.
\end{enumerate}
\eec
\begin{remark}  The result immediately generalizes to $A\in \mfb^{S_*}$, cf. Remark~\ref{Span}.
\end{remark}
\section{Scattering of photons in the vacuum sector} \label{vacuum} \setcounter{equation}{0}

In this section we consider a Haag-Kastler net $(\mfa,U)$ in a vacuum representation containing
massless Wigner particles (`photons'). We collect here some basic facts about asymptotic fields of 
photons in a vacuum representation, which will be needed in our discussion of charged representations in Section~\ref{charged}.
These results were first established in \cite{Bu77} and recently revisited in \cite{DH15}, were  simpler proofs,
exploiting energy bounds, were given. In another recent work  asymptotic fields were constructed for all $A\in \mfb_{\loc,0}$ \cite{Ta14},
but we will not explore this direction here and content ourselves with $A\in \mfb^{C_*}$.
Instead, we take the opportunity to indicate another simplification. Namely, in the analysis of commutators of asymptotic fields in Proposition~\ref{commutator-proposition} below, the clustering estimates from \cite{Bu77, DH15} can be avoided due to Theorem~\ref{main-ergodic}.

Since the approximating sequences $\bar{A}_t\{f\}$ we use  are  different than in \cite{Bu77, DH15}, we give 
rather complete proofs. We start our discussion with the following standard lemma.
\begin{lemma}\label{StrongVacuum}
Let $A\in\mfb^{C_*}$ and $f\in C^{\infty}(S^2)$. Then,
\begin{equation}
\lim_{t\to\infty}\bar{A}_t\{f\}\Om=P_{\pho} f\left(\tfrac{\mathbf{P}}{|\mathbf{P}|}\right)A\Om, \label{single-particle-photon}
\end{equation}
where $P_{\pho}$ is the projection onto the subspace of massless one-particle states $\mfh_{\pho}$. Vectors
on the right-hand side of (\ref{single-particle-photon}) span a dense, Poincar\'e invariant subspace $D_{\pho}$ in $\mfh_{\pho}$. 
(The subspace $D_{\pho}^{(1)}\subset D_{\pho}$, spanned by vectors with $f\equiv 1$, is also dense and Poincar\'e invariant).
\end{lemma}
\proof It suffices to prove the lemma for $A\in\mfb_{C_*}$ and then extend by linearity.  By Lemma~\ref{spherical-means-analysis}, we have
$A_t\{f\}=\pa_0B(g\ast_3 f_t)(t)$, $B\in \mfb_{C_*}$ and $g\in C_*(\mathbb{R}^4)$. 
Thus, we have
\beqa
A_t\{f\}\Om=\pa_0B(g\ast_3 f_t)(t)\Om=(2\pi)^{2}\e^{\i tH}(\wt{g\ast_3 f_t})(P)\pa_0B\Om.
\eeqa
Making now use of formula~(\ref{convolution-formula}), Theorem~\ref{MWLem} and  Proposition~\ref{Ergodic-corollary} (a),
we obtain
\begin{align}
\lim_{t\to\infty}\bar{A}_t\{f\}\Om&=\lim_{t\to\infty}(2\pi)^{2}\int dt'\, h_t(t') \e^{\i t'H}(\wt{g\ast_3 f_{t'}})(P)\pa_0B\Om\non\\
&=(2\pi)^{2}E(\pa\ov{V}_+)  \fr{\wt{g}(P)}{\i|\mathbf{P}|} f\left(\tfrac{\mathbf{P}}{|\mathbf{P}|}\right)\i HB\Om\non\\
&=(2\pi)^{2} P_{\pho}  \wt{g}(P) f\left(\tfrac{\mathbf{P}}{|\mathbf{P}|}\right) B\Om\non\\
&=P_{\pho}  f\left(\tfrac{\mathbf{P}}{|\mathbf{P}|}\right) A\Om.
\end{align}
Here we used that $HB\Om$ is in the domain of $|\bP|^{-1}$, as one can show using the JLD method \cite[p.149]{Bu77}.
We also exploited that $HB\Om$ is orthogonal to the vacuum and thus $E(\pa\ov{V}_+)$ can be replaced with $P_{\pho}$. 

Poincar\'e invariance of $D_{\pho}$ follows from the relation
\beqa
\lim_{t\to\infty}U(\la)\bar{A}_t\{f\}\Om=P_{\pho}  
(f\circ g_{\La^{-1}})\left(\tfrac{\mathbf{P}}{|\mathbf{P}|}\right) A_{\la}\Om
=\lim_{t\to\infty}\bar{A}_{\la,t}\{f\circ g_{\La^{-1}}\}\Om, \label{Poincare-computation}
\eeqa
where $A_{\la}:=U(\la)AU(\la)^*\in \mfb_{C_*}$ and (\ref{transform}) was taken into account.

To show density, we exploit the cyclicity of the vacuum under $\mfb$ and the fact that 
with functions $\wt g$, where $g\in C_*(\real^4)$, one can approximate pointwise the characteristic function of $\real^{4}\backslash \{0\}$.\hfill \qed

Next, denote by $\mathcal{O}_{+}$ the future tangent of a double cone $\mathcal{O}$, i.e. the  cone of all points that have a 
positive timelike separation from  $\mathcal{O}$. Following the arguments of \cite{Bu77}, based on the Huygens principle,  we have:
\begin{lemma}\label{lem2}
Let $A\in\mfb^{C_*}(\mco)$ and $f\in C^{\infty}(S^2)$. Then, the limit
\begin{equation}
A^{\out}\{f\}\Psi=\lim_{t\to\infty}\bar{A}_t\{f\}\Psi \label{first-A-out}
\end{equation}
exists for $\Psi$ in the dense domain $D(\mco):=\{\, B\Omega \,|\,
 B\in\mfb_{\loc}(\mco_{+})\, \}$.
Moreover, $A^{\out}\{f\}$ depends only on the single-particle state $A^{\out}\{f\}\Om$ within
the above restrictions. 
\end{lemma}
\begin{proof} Let $A$ be localized in  $\mathcal{O}$ and $\supp\, f$ be contained in the set 
$\Theta\subset S^2$.  Then, by construction, $\bar{A}_t\{f\}$ is localized in the region 
\beqa\label{Ot}
\mco_{t}:=\bigcup\limits_{\tau\in t+t^{\beps}\supp h}\left\{\mathcal{O}+\tau(1,\Theta)\right\}. 
\eeqa
Clearly, for sufficiently large $t$ the  region $\mco_{t}$ is  spacelike separated
from any given double cone $\mathcal{O}_{1}$ in $\mathcal{O}_{+}$. Thus, it follows from
Lemma~\ref{StrongVacuum} and
the locality property that for all $B\in \mfb_{\loc}(\mco_{+})$, 
\begin{equation}
\lim_{t\to\infty}\bar{A}_t\{f\}B\Omega=
\lim_{t\to\infty}B\bar{A}_t\{f\}\Omega=B\,P_{\pho}\, f\left(\tfrac{\mathbf{P}}{|\mathbf{P}|}\right)A\Om, \label{dense-dom-computation}
\end{equation}
defining $A^{\out}\{f\}$ on the domain $D(\mco)$. This domain is dense as shown in \cite{Bu75}.

It is manifest from the above discussion that given $\bar{A'}_t\{f'\}$, where $A'\in\mfb^{C_*}(\mco)$ and $f'\in C^{\infty}(S^2)$, such that $A^{\out}\{f\}\Om=A'^{\out}\{f'\}\Om$ we have $A^{\out}\{f\}=A'^{\out}\{f'\}$ as operators on
$D(\mco)$.\hfill \qed
\end{proof}
\nin In view of Lemmas~\ref{StrongVacuum} and \ref{lem2}, and Proposition~\ref{uniform-bound-prop} we obtain the following result.
\bep\label{admissible-two} 
Let $A, A_i\in\mfb^{C_*}$ and $f, f_i\in C^{\infty}(S^2)$, $i=1,\ldots,n$. Then:
\begin{enumerate}
\item[(a)] For any $\Psi\in D_H$ the limit $\lim_{t\to\infty}\bar{A}_t\{f\}\Psi$ exists and defines a closable operator $A^{\out}\{f\}\res D_H$.
This operator is uniquely determined by  the vector $A^{\out}\{f\}\Om$. 
\item[(b)]  $A^{\out}\{f\} D_H\subset D_H$. 
\item[(c)] $A_1^{\out}\{f_1\}\ldots A_n^{\out}\{f_n\}\Psi=\lim_{t\to\infty}\bar{A}_{1,t}\{f_1\}\ldots \bar{A}_{n,t}\{f_n\}\Psi$ for $\Psi\in D_H$.
\end{enumerate}
\eep
The next lemma settles the transformation rules of the asymptotic fields under the  Poincar\'e transformations.
\bel\label{Poincare-lemma} Let $A\in\mfa^{C_*}(\mco)$ and $f\in C^{\infty}(S^2)$.  For any $\la\in \wt{\mathcal{P}}_+^{\uparrow}$ we have
on $D_H$
\beqa
U(\la)A^{\out}\{f\}U(\la)^*=A_{\la}^{\out}\{f\circ g_{\La^{-1}}\}, \label{Lorentz-covariance}
\eeqa
where $A_{\la}:=U(\la)AU(\la)^*\in\mfa^{C_*}(\la\mco)$ and $g_{\La}$ was defined in  (\ref{g-Lambda}).
\eel
\proof  Recall that $D(\mco):=\{ B\Om\,|\, B\in \mfa_{\loc}(\mco_+) \,\}$ and note the relation
$D(\mco)=U(\la)^*D(\la\mco)$. By formula~(\ref{Poincare-computation}) and Lemma~\ref{lem2} we
obtain that  relation~(\ref{Lorentz-covariance}) holds on $D(\la\mco)$. Next, we choose $\Phi, \Psi\in D_H$
and $\Psi_n\in D(\la\mco)$ such that $\|\Psi-\Psi_n\|\leq 1/n$. Then, making use of the fact that 
$U(\la)^*D_H\subset D_H$ and the observation that $\Phi$ is in the intersection of domains of $(U(\la)A^{\out}\{f\}U(\la)^*)^*$ 
and  $A_{\la}^{\out}\{f\circ g_{\La^{-1}}\}^*$,
we have
\beqa
\lan \Phi, U(\la)A^{\out}\{f\}U(\la)^*\Psi\ran=\lan \Phi, A_{\la}^{\out}\{f\circ g_{\La^{-1}} \}\Psi\ran+O(1/n).
\eeqa
Keeping $\Phi$ and $\Psi$ fixed, we can take the limit $n\to\infty$ and drop the error term. The claim follows
from the resulting relation. \hfill \qed\\
Now we analyse  commutators of asymptotic fields. We proceed similarly as in the case of massless fermions \cite{Bu75}. 
Till the end of this section $A, A', A_i\in \mfa^{C_*}$ and $f, f', f_i\in C^{\infty}(S^2)$ unless stated otherwise.
\bel\label{first-commutator-lemma} Let $A\in \mfa^{C_*}(\mco)$. Then, for all $B\in\mfb_{\loc}(\mco_+)$
\beqa
[A^{\out}\{f\}, B]=0. \label{first-comm-equality}
\eeqa 
 Moreover, if $A'\in \mfa^{C_*}(\mco')$, then
\beqa
[A^{\out}\{f\}, A'^{\out}\{f'\}(x)]=0 \label{second-comm-equality}
\eeqa
provided that $\mco'+x\subset\mco_+$. Both equalities hold in the sense of quadratic forms on $D_H\times D_H$.
\eel
\proof In view of Proposition~\ref{admissible-two} we can write for any $\Psi, \Phi\in D_H$
\beqa
\lan \Psi,[A^{\out}\{f\}, B]\Phi \ran=\lim_{t\to\infty}\lan \Psi,[\bar{A}_t\{f\}, B]\Phi \ran=0,
\eeqa
since the commutator vanishes for sufficiently large $t$. By approximating
\beqa
\lan \Psi, [A^{\out}\{f\}, A'^{\out}\{f'\}(x)]\Phi\ran=
\lim_{t\to\infty}\lan \Psi, [A^{\out}\{f\}, \bar{A'}_t\{f'\}(x)]\Phi\ran,
\eeqa
and noting that $\bar{A'}_t\{f'\}(x)$ is localized in  $\mco_+$ for sufficiently large~$t$, we obtain relation~(\ref{second-comm-equality})
from~(\ref{first-comm-equality}) \hfill \qed
\bel\label{wave-equation} $x\mapsto A^{\out}\{f\}(x)$ is a solution of the wave equation. 
That is, 
\beqa
\square_x A^{\out}\{f\}(x)\Psi=0 \quad\textrm{for}\quad  \Psi\in D_H.
\eeqa
\eel
\proof It follows immediately from Lemma~\ref{StrongVacuum}, that  $\square_x A^{\out}\{f\}(x)\Om=0$.
Hence, by Proposition~\ref{admissible-two} (a),   $\square_x A^{\out}\{f\}(x)\Psi=0$
for any $\Psi\in D_H$. \hfill \qed
\bep\label{commutator-proposition} Let $A^{\out}\{f\}$, $A'^{\out}\{f'\}$ be two asymptotic fields as specified above. Then,
\beqa
\,[A^{\out}\{f\}, A'^{\out}\{f'\}]=\lan \Om,[A^{\out}\{f\}, A'^{\out}\{f'\}]\Om\ran 1_{\hil}
\eeqa
as operators on $D_H$. 
\eep
\proof First, we use a method of Pohlmeyer \cite{Po69} (applied also in the collision theory of massless fermions \cite{Bu75}) to show that
\beqa
\,[A^{\out}\{f\}, A'^{\out}\{f'\}]\Om=c\,\Om,\qquad c\in\mathbb{C}. \label{vacuum-reproduction}
\eeqa
To this end, we take any vector $\Phi$ such that $\Phi=E(K_{\Phi})\Phi$ for a compact set
$K_{\Phi}$ in the interior of the future light cone and consider the function
\beqa
F(x,y)=\lan \Phi, [A^{\out}\{f\}(x), A'^{\out}\{f'\}(y)]\Om\ran.
\eeqa
Making use of Lemma~\ref{wave-equation} and the energy-momentum transfer relation~(\ref{EM-transfer}), we get that the support of the Fourier
transform of $F$ is contained in the compact set
\beqa
\{\, p,q\in \real^4 \,|\, p_0^2=|\bp|^2, \, q_0^2=|\mathbf{q}|^2, p+q\in K_{\Phi}\,\}.
\eeqa
Therefore, $F$ is an entire analytic function and since it vanishes on  an open subset of 
$\real^8$ by Lemma~\ref{first-commutator-lemma}  it vanishes everywhere. Hence,
\beqa
\,[A^{\out}\{f\}, A'^{\out}\{f'\}]\Om=c\,\Om+\Psi_{\pho},
\eeqa
where $\Psi_{\pho}\in \mfh_{\pho}$. Thus to prove~(\ref{vacuum-reproduction}) it remains 
to show that $\Psi_{\pho}=0$ \footnote{In  collision theory of massless fermions  $\Psi_{\pho}=0$ was automatic in the corresponding expression, since a bosonic operator cannot create a fermionic single-particle state from the vacuum \cite[Lemma~4]{Bu75}. In the present bosonic case we can conclude using Theorem~\ref{main-ergodic}.}. For this purpose, we choose $\Phi_1, \Phi_2\in \mfh_{\pho}\cap D_H$  and compute by means of Proposition~\ref{creation-annihilation} and Corollary~\ref{main-ergodic-corollary} that
\beqa
\lan \Phi_1, A^{\out}\{f\}\Phi_2\ran=\lan \Phi_1, A^{\out}\{f\}^+\Phi_2\ran+\lan \Phi_1, A^{\out}\{f\}^-\Phi_2\ran=0.
\eeqa

Given (\ref{vacuum-reproduction}), we complete the proof of the proposition as follows.
Let $\mco$ be a double cone such that $A, A'\subset \mfa^{C_*}(\mco)$.
Then, for any $B\in \mfb_{\loc}(\mco_+)$ and $\Psi\in D_H$
\begin{align}
\lan \Psi, [A^{\out}\{f\}, A'^{\out}\{f'\}]B\Om\ran
=&\lim_{t\to\infty} \lan \Psi, [\bar{A}_t\{f\}, \bar{A}'_t\{f'\}]B\Om\ran\non\\
=&\lim_{t\to\infty}\lan \Psi, B[\bar{A}_t\{f\}, \bar{A}'_t\{f'\}]\Om\ran\non\\
=&\lan \Psi, B\Om\ran \lan \Om, [A^{\out}\{f\},  A'^{\out}\{f'\}]\Om\ran, \label{commutator-final-step}
\end{align}
where in the first step Proposition~\ref{admissible-two} (c) and in the second step the localization properties of the approximating sequences entered. Equation~(\ref{commutator-final-step}) extends by continuity from
$D(\mco_+)$ to $D_H$, since $\Psi$ is  in the domain
of $ ([A^{\out}\{f\}, A'^{\out}\{f'\}])^*$.\hfill \qed 

With the above proposition we have all the necessary ingredients for our discussion of Compton
scattering in Section~\ref{charged}. For the sake of completeness, however, we indicate below the construction
of scattering states of photons in the vacuum sector. To this end, it suffices to consider functions $f\in C^{\infty}(S^2)$
which are identically equal to one, in which case we simply write $A^{\out}$ for $A^{\out}\{f\}$.  

From Proposition~\ref{commutator-proposition} we immediately obtain the canonical commutation
relations for the asymptotic creation and annihilation operators, namely
\beqa
\,[A^{\out-}, A'^{\out+}]=\lan A^{*\out+}\Om, A'^{\out+} \Om\ran,\quad [A^{\out-}, A'^{\out-}]=[A^{\out+}, A'^{\out+}]=0. \label{commutation}
\eeqa
Recall that by Proposition~\ref{creation-annihilation} (c) we have  $A^{\out+}D_H\subset D_H$. Hence, scattering states may be constructed in a standard manner.
Parts (a) and (b) of the following theorem are a direct consequence of (\ref{commutation}) and part (c) is proven analogously 
as Theorem~\ref{Compton-scattering}~(c) stated below. 
\bet\label{photon-scattering} The states $\Psi^{\out}:= A_1^{\out+}\ldots A_n^{\out+}\Om$ have the following properties:
\begin{enumerate}
\item[(a)] $\Psi^{\out}$ depends only on the single-particle states $\Phi_i=A_i^{\out}\Om\in D_{\pho}$. Therefore, we put $\Psi^{\out}=\Phi_1\tout\cdots \tout\Phi_n$.
\item[(b)] 
$\lan  \Phi_1\tout\cdots \tout\Phi_n,\Phi'_1\tout\cdots \tout\Phi'_{n'}\ran=
\de_{n,n'}\sum_{\si\in \mathfrak{S}_n}\lan \Phi_1, \Phi'_{\si_1}\ran\ldots \lan \Phi_n, \Phi'_{\si_n}\ran$, where
$\mathfrak{S}_n$ is the set of all permutations of $(1,\ldots, n)$.
\item[(c)] $\wtU(\wtla)(\Phi_1\tout\cdots \tout\Phi_n)=(\wtU(\wtla)\Phi_1)\tout\cdots \tout (\wtU(\wtla)\Phi_n)$, where $\wtla\in\widetilde{\mathcal{P}}_+^\uparrow$.
\end{enumerate}
\eet

\section{Compton scattering in hypercone localized representations}\label{charged} \setcounter{equation}{0}

In this section we consider a Haag-Kastler net $(\mfa, U)$ in a vacuum representation, containing massless Wigner particles (`photons')
and a  representation $\pi$ which is hypercone localized w.r.t. this vacuum and describes
Wigner particles of mass $m>0$ (`electrons').  For brevity
we will write $(\hat\mfa, \hat U)$ for the resulting net $(\mfa_{\pi}, U_{\pi})$ and $\hat\hil:=\hil_{\pi}$. We also
set $\hat A:=\pi(A)$ for $A\in \mfa$ and denote by  ($\hat H, \hat \bP$) the energy-momentum
operators in the representation $\pi$.

Given $\hA\in  \hat \mfb^{C_*}(\mco)$ and $\supp\, f\subset \Theta\subset S^2$, the asymptotic field approximants  $t\mapsto \bar{\hA}_t\{f\}$ are localized in
\beqa
\mco_{t}:=\bigcup\limits_{\tau\in t+t^{\beps}\supp h}\left\{\mathcal{O}+\tau(1,\Theta)\right\}, \quad t\geq 1. \label{mco-t}
\eeqa
In the course of our analysis we will also consider  $t\mapsto \bar{\hA}_{\la,t}\{f\circ g_{\La^{-1}}\}$, $\la=(0,\wt\La)\in \wt{\mathcal P}_+^{\uparrow}$, whose
localization regions are
\beqa
\mco_{t}^{\La}:=\bigcup\limits_{\tau\in t+t^{\beps}\supp h}\left\{\La\mathcal{O}+\tau(1,  g_{\La}(\Theta))\right\}, \quad t\geq 1. \label{mco-t-one}
\eeqa
The following geometric lemma will be frequently used  in the subsequent discussion. Its proof can be found in Appendix~\ref{geometric}.
\bel\label{geometric-one} For any $\mco\in \mcK$ and any open $\Theta\subset S^2$ such that $\ov{\Theta}\subsetneq S^2$ 
there is a future lightcone $V$, a hypercone $\mcC\subset\mcF_V$ and a neighbourhood $N$ of unity in the Lorentz group such that 
\beqa
\La O_t\subset \mcC^{\cc}, \quad   O^{\La}_t\subset \mcC^{\cc}, \quad t\geq 1,
\eeqa
for all $\La\in N$. 
\eel
Given Lemma~\ref{geometric-one}, the existence of a certain family of asymptotic fields is easily obtained. The following result holds true.
\bel\label{charge-class-lem2}  Let $\hA\in  \hat \mfb^{C_*}$, $f\in C^{\infty}(S^2)$ and $\supp\, f \subset \Theta$, 
with $\Theta$ as in Lemma~\ref{geometric-one}.  
Then, the limit
\beqa
\hA^{\out}\{f\}\Psi:=\lim_{t\to\infty} \bar{\hA}_{t}\{f\}\Psi, \quad \Psi\in D_{\hat H}
\eeqa
exists. It defines a closable operator on $D_{\hat H}$ which is uniquely determined by $A^{\out}\{f\}\Om$. 
\eel
\proof Let $\mco$ be the localization region of $\hA$ and  $\mco_t$ be given by (\ref{mco-t}). 
By  Lemma~\ref{geometric-one}, there exists a future lightcone $V$ and $\mcC\subset\mcF_{V}$ such that
$O_t\subset \mcC^{\cc}$. Hence, by hypercone localization of $\pi$, there is a unitary $W_{\mcC}$ such that
for all $t\geq 1$
\beqa
\bar{\hA}_t\{f\}=\pi(\bar{A}_t\{f\})=W_{\mcC}(\bar{A}_t\{f\})W_{\mcC}^*.
\eeqa
Now by Proposition~\ref{admissible-two} the right-hand side converges on  $W_{\mcC}D_H$
to an operator which is uniquely determined by $A^{\out}\{f\}\Om$. Then, by Proposition~\ref{general-admissible-two},
the left-hand side converges on $D_{\hat H}$ to an operator which is uniquely determined by $A^{\out}\{f\}\Om$.  \hfill \qed\\
In the next proposition we eliminate the restriction on functions $f$.
\bep\label{charge-class-admissible-two} Let $\hA, \hA_i\in\hat{\mfb}^{C_*}$ and  $f, f_i\in C^{\infty}(S^2)$. Then:
\begin{enumerate}
\item[(a)] For any $\Psi\in D_{\hat H}$ the limit $\lim_{t\to\infty}\bar{\hA}_t\{f\}\Psi$ exists and defines a closable operator $\hA^{\out}\{f\}$ on $D_{\hat H}$ which is uniquely specified by $A^{\out}\{f\}\Om$. 
\item[(b)]  $\hA^{\out}\{f\} D_{\hat H}\subset D_{\hat H}$. 
\item[(c)] $\hA_1^{\out}\{f_1\}\ldots \hA_n^{\out}\{f_n\}\Psi=\lim_{t\to\infty}\bar{\hA}_{1,t}\{f_1\}\ldots \bar{\hA}_{n,t}\{f_n\}\Psi$ for $\Psi\in D_{\hat H}$.
\end{enumerate}
\eep
\begin{remark} It follows immediately from Proposition~\ref{charge-class-admissible-two} (a) that $x\mapsto \hA^{\out}\{f\}(x)\Psi$, $\Psi\in D_{\hat H}$,
is a solution of the wave equation. 
\end{remark}
\proof In view of Lemma~\ref{charge-class-lem2} and  Proposition~\ref{general-admissible-two}, we obtain the statement of the proposition for $f, f_i$ supported in proper open subsets
of $S^2$. 
To remove this restriction, we choose a  partition of unity on $S^2$ consisting
of $f^j\in C^{\infty}(S^2)$, $j=1,2$,  such that $\supp\, f_j\subsetneq S^2$. Thus, we may write
\beqa
\bar{\hA}_{t}\{f\}=\sum_{j=1,2} \bar{\hA}_{t}\{ff^j\}.
\eeqa
Now it is easy to see that $\hat A^{\out}\{f\}=\lim_{t\to\infty}\bar{\hA}_{t}\{f\}$ exists on $D_{\hat H}$ and  
has the properties specified in the proposition. The only property which requires an argument is
the last statement in part (a). To verify it, suppose that $A^{\out}\{f\}\Om=A^{'\out}\{f'\}\Om$. Then, we also have $A^{\out}\{ff^j\}\Om=A^{'\out}\{f'f^j\}\Om$, since
$ A^{\out}\{f\}\Om=P_{\pho}f(\bP/|\bP|)A\Om$. Therefore, by Lemma~\ref{charge-class-lem2}, we have that $\hA^{\out}\{ff^j\}=\hA^{'\out}\{f'f^j\}$ on $D_{\hat H}$.
Summation over $j$ yields the claim.
\hfill \qed\\
Next, we analyze the transformation rules of the asymptotic fields under Poincar\'e transformations. Our result is as follows.
\bel\label{Poincare-charged}  Let $\hA\in\hat\mfa^{C_*}(\mco)$ and $f\in C^{\infty}(S^2)$.  For any $\la\in \wt{\mathcal{P}}_+^{\uparrow}$ we have
on $D_{\hat H}$
\beqa
\hat U(\la)\hat A^{\out}\{f\}\hat U(\la)^*=\hA_{\la}^{\out}\{f\circ g_{\La^{-1}}\}, \label{Lorentz-covariance-charged}
\eeqa
where $\hA_{\la}:=\hat U(\la)\hA \hat U(\la)^*\in\hat\mfa^{C_*}(\la\mco)$ and $g_{\La}$ is given by (\ref{g-Lambda}).
\eel
\proof  First, decompose $f$ according to 
$f=\sum_{j=1,2} f^j$,
with  $f^j\in C^{\infty}(S^2)$,  $\supp\, f^j\subset \Theta_j$ and $\Theta_j$ as in Lemma~\ref{geometric-one}. Next, choose $\wt N\subset SL(2,\complex)$ such that its image in the Lorentz group under the canonical covering map is contained in the neighbourhood $N$ from Lemma~\ref{geometric-one}. (We can find one $N$ for both values of $j$).  Now for any $j$ Lemma~\ref{geometric-one} gives a  future lightcone $V_j$ and a hypercone $\mcC_j\in \mcF_{V_j}$ 
such that
\beqa
U(\la)\bar A_{t}\{f^j\}U(\la)^*\in \mfa(\mcC_{j}^{\cc}),\quad 
\bar A_{\la, t}\{f^j\circ g_{\La^{-1}}\} \in \mfa(\mcC_{j}^{\cc})
\eeqa
for all $t\geq 1$ and $\la=(0,\wt\La)$, $\wt\La\in \wt N$. Thus, due to the Poincar\'e covariance and  hypercone localization, we have unitaries $W_{\mcC_j}$ such that
\begin{align}
\hat U(\la)\bar \hA_{t}\{f^j\}\hat U(\la)^*&=W_{\mcC_j}\big(U(\la)\bar A_{t}\{f^j\}U(\la)^*\big)W_{\mcC_j}^*, \label{hypercone-poincare-first}\\
\bar \hA_{\la, t}\{f^j\circ g_{\La^{-1}}\}&=W_{\mcC_j} \bar A_{\la, t}\{f^j\circ g_{\La^{-1}}\} W_{\mcC_j}^*. 
\label{hypercone-poincare-second}
\end{align}
It follows from the above relations and Lemma~\ref{lem2} that 
\beqa
W_{\mcC_j} D(\la \mco)\subset \hat U(\la) D_{\max}(\hA, f^j)\cap  D_{\max}(\hA_{\la}, f^j\circ g_{\La^{-1}}).  
\eeqa
Hence, both sides of
\beqa
\hat U(\la)\hat A^{\out}\{f^j\}\hat U(\la)^*\Psi=\hA_{\la}^{\out}\{f^j\circ g_{\La^{-1}}\}\Psi, \quad \Psi\in W_{\mcC_j}D(\la \mco)
\label{charged-poincare-comp}
\eeqa
are well defined.
To verify equality~(\ref{charged-poincare-comp}), we choose $B\in \mfa_{\loc}(\mco_+)$, set $B_{\la}:=U(\la)BU(\la)^*$ and compute
\begin{align}
\hat U(\la)\hat A^{\out}\{f^j\}\hat U(\la)^*W_{\mcC_j}B_{\la}\Om
&=\lim_{t\to\infty} \hat U(\la)\bar{\hat A}_t\{f^j\}\hat U(\la)^*W_{\mcC_j}B_{\la}\Om\non\\
&=\lim_{t\to\infty} W_{\mcC_j} U(\la)\bar{A}_t\{f^j\}U(\la)^*B_{\la}\Om\non\\
&=\lim_{t\to\infty} W_{\mcC_j} U(\la)\bar{A}_t\{f^j\}B\Om\non\\
&=\lim_{t\to\infty} W_{\mcC_j} \bar{A}_{\la,t}\{f^j\circ g_{\La^{-1}}\}B_{\la}\Om\non\\
&=\hA_{\la}^{\out}\{f^j\circ g_{\La^{-1}}\}  W_{\mcC_j} B_{\la}\Om,
\end{align}
where in the second step we used (\ref{hypercone-poincare-first}), in the fourth step~(\ref{Poincare-computation}) and in the
last step (\ref{hypercone-poincare-second}). Arguing as in the proof of Lemma~\ref{Poincare-lemma}, we conclude
from (\ref{charged-poincare-comp}) that 
\beqa
\hat U(\la)\hat A^{\out}\{f^j\}\hat U(\la)^*\Psi=\hA_{\la}^{\out}\{f^j\circ g_{\La^{-1}}\}\Psi, \quad \Psi\in D_{\hat H}.
\label{charged-poincare-comp-D-H}
\eeqa
Summing the above relation over $j$ we obtain the claim for $\la=(0,\wt\La)$, $\wt\La\in \wt N$.

It remains to extend the result to arbitrary  $\la\in \wt{\mathcal{P}}_+^{\uparrow}$. To this end, we first note that (\ref{Lorentz-covariance-charged}) holds trivially for $\la=(x, I)$. Now 
any element of $\wt{\mathcal{P}}_+^{\uparrow}$  can be written as $(x, I)(0,\wt \La)$, $\wt\La\in SL(2,\complex)$. Since $SL(2,\complex)$
is connected, it is generated by any neighbourhood of the identity. \hfill\qed

Exploiting the hypercone localization and Proposition~\ref{commutator-proposition}, we obtain that commutators of
asymptotic fields are numbers.
\bep\label{commutator-proposition-new} Let $\hA_1, \hA_2\in\hat{\mfb}^{C_*}$ and $f_1, f_2\in C^{\infty}(S^2)$. Then,
\beqa
\, [\hA_1^{\out}\{f_1\}, \hA_2^{\out}\{f_2\}]=\lan \Om, [A_1^{\out}\{f_1\}, A_2^{\out}\{f_2\}]\Om \ran  1_{\hat\hil},
\eeqa
as operators on $D_{\hat H}$.
\eep
\proof  We decompose
$f_{i}=\sum_{j=1}^{\ell} f_{i}^{j}$,  $i=1,2$,
where $f_{i}^{j}\in C^{\infty}(S^2)$ are supported in sufficiently small subsets of $S^2$. Thus, we have
\beqa
\, [\bar{\hA}_{1,t}\{f_1\}, \bar{\hA}_{2,t}\{f_2\}]=\sum_{j_1,j_2}[\bar{\hA}_{1,t}\{f_1^{j_1}\}, \bar{\hA}_{2,t}\{f_2^{j_2}\}].
\eeqa
We divide the set of indices into two subsets, namely
\beqa
& &S:=\{\, (j_1, j_2)\,|\, \supp\, f_{1}^{j_1}\cap  \supp\, f_{2}^{j_2}\neq\emptyset\,\},\\
& &S':=\{\, (j_1, j_2)\,|\, \supp\, f_{1}^{j_1}\cap  \supp\, f_{2}^{j_2}=\emptyset\,\}.
\eeqa
If the partition is sufficiently fine, for any $(j_1, j_2)\in S$ Lemma~\ref{geometric-one} gives a future lightcone $V_{j_1,j_2}$ 
and a hypercone $\mcC_{j_1,j_2}\in \mcF_{V_{j_1,j_2}}$ such that 
\beqa
\, [\bar{A}_{1,t}\{f_1^{j_1}\}, \bar{A}_{2,t}\{f_2^{j_2}\}]\in \mfa( \mcC_{j_1,j_2}^{\cc} )
\eeqa
for all $t\geq 1$. Thus, by the hypercone localization of $\pi$ there is a unitary $W_{j_1,j_2}$ such that
\beqa
\, [\bar{\hA}_{1,t}\{f_1^{j_1}\}, \bar{\hA}_{2,t}\{f_2^{j_2}\}]=W_{j_1,j_2}[\bar{A}_{1,t}\{f_1^{j_1}\}, \bar{A}_{2,t}\{f_2^{j_2}\}]W_{j_1,j_2}^*.
\eeqa
Now for $\Psi\in D_{\hat H}$ and $\Phi\in  W_{j_1,j_2} D_H$ we have by Propositions~\ref{admissible-two}, \ref{commutator-proposition} and
\ref{charge-class-admissible-two} that
\begin{align}
\lan\Psi, [\hA^{\out}_{1}\{f_1^{j_1}\}, \hA^{\out}_{2}\{f_2^{j_2}\}]\Phi\ran
&=\lim_{t\to\infty}\lan\Psi, [\bar{\hA}_{1,t}\{f_1^{j_1}\}, \bar{\hA}_{2,t}\{f_2^{j_2}\}]\Phi\ran\non\\
&=\lim_{t\to\infty}\lan \Psi, W_{j_1,j_2}   [\bar{A}_{1,t}\{f_1^{j_1}\}, 
\bar{A}_{2,t}\{f_2^{j_2}\}]    W_{j_1,j_2}^*\Phi\ran\non\\
&=\lan \Psi,  \Phi\ran \lan \Om,   [A^{\out}_{1}\{f_1^{j_1}\}, A^{\out}_{2}\{f_2^{j_2}\}] \Om\ran.
\end{align}
Since $\Psi$ is in the domain of $([\hA^{\out}_{1}\{f_1^{j_1}\}, 
\hA^{\out}_{2}\{f_2^{j_2}\}])^*$, the above equality extends to $\Phi\in D_{\hat H}$.

Finally, we consider $(j_1, j_2)\in S'$. In this case locality gives  for sufficiently large $t$
\beqa
[\bar{\hA}_{1,t}\{f_1^{j_1}\}, \bar{\hA}_{2,t}\{f_2^{j_2}\}]=0,
\eeqa
as one can see by a straightforward computation. This concludes the proof.\hfill \qed\\
After these preparations we proceed to the construction of scattering states of one electron and a finite number of photons i.e. Compton scattering. It suffices to consider $f\in C^{\infty}(S^2)$ which are identically equal to one, in which case we write, as in the previous section, $\hA^{\out}$ for $\hA^{\out}\{f\}$. Similarly as in the vacuum representation, Proposition~\ref{commutator-proposition-new} gives
\beqa
\,[\hA^{\out-}, \hA'^{\out+}]=\lan A^{*\out+}\Om, A'^{\out+} \Om\ran,\quad [\hA^{\out-}, \hA'^{\out-}]=[\hA^{\out+}, \hA'^{\out+}]=0. \label{commutation-two}
\eeqa
Recalling that by Proposition~\ref{creation-annihilation} (c)  $\hA^{\out+}D_{\hat H}\subset D_{\hat H}$, scattering states are constructed in a straightforward manner.
\bet\label{Compton-scattering} The states $\Psi^{\out}:= \hA_1^{\out+}\ldots \hA_n^{\out+}\Psi_{\el}$, $\Psi_{\el}\in \hat\mfh_{\el}\cap D_{\hat H}$, have the following properties:
\begin{enumerate}
\item[(a)] $\Psi^{\out}$ depends only on the single-photon states $\Phi_i:=A_i^{\out}\Om\in D_{\pho}$ and the single-electron state $\Psi_{\el}\in \hat\mfh_{\el}\cap D_{\hat H}$. Therefore, we can
 write $\Psi^{\out}=\Phi_1\tout\cdots \tout\Phi_n\tout\Psi_{\el}$.
\item[(b)] Given $\Psi^{\out}, \Psi^{'\out}$ as above, 
\beqa
\lan \Psi^{\out}, \Psi^{'\out}\ran=\de_{n,n'}\lan \Psi_{\el},\Psi_{\el}'\ran
\sum_{\si\in \mathfrak{S}_n}\lan \Phi_1, \Phi'_{\si_1}\ran\ldots \lan \Phi_n, \Phi'_{\si_n}\ran, 
\eeqa
where
$\mathfrak{S}_n$ is the set of all permutations of $(1,\ldots, n)$.
\item[(c)] $\hat\wtU(\wtla)(\Phi_1\tout\cdots \tout\Phi_n\tout \Psi_{\el})=(\wtU(\wtla)\Phi_1)\tout\cdots \tout (\wtU(\wtla)\Phi_n)\tout(\hat\wtU(\wtla)\Psi_{\el})$,  $\wtla\in\widetilde{\mathcal{P}}_+^\uparrow$.
\end{enumerate}
\eet
\proof Parts (a) and (b) follow directly from (\ref{commutation-two}) and Corollary~\ref{main-ergodic-corollary} which gives $\hA^{\out-}\Psi_{\el}=0$ for $\Psi_{\el}\in \hat\mfh_{\el}\cap D_{\hat H}$.   To prove (c), it suffices to consider $\hat A_i\in \hat\mfa_{C_*}$,
the general case follows by linearity (cf. Remark~\ref{Span}). We use Lemma~\ref{Poincare-charged} and the following computation
\begin{eqnarray}
\hspace{-.5cm}&&\hspace{-.5cm}\hat\wtU(\wtla)
\hA_1^{\out}(\eta_{+,r_1})\ldots \hA_n^{\out}(\eta_{+,r_n})\Psi_{\el}\non\\
&=&\int \prod_{j=1}^{n}d^4x_j (\eta_{+,r_1}\otimes\cdots\otimes\eta_{+,r_n})(x_1,\dots,x_n) \hat\wtU(\wtla)\hA_1^{\out}(x_1)\hA_2^{\out}(x_2)\ldots \hA_n^{\out}(x_n)\Psi_{\el}\non\\
&=&\int \prod_{j=1}^{n}d^4x_j (\eta_{+,r_1}\otimes\cdots\otimes\eta_{+,r_n})(x_1,\dots,x_n)
 \hA_{1,\la}^{\out}(\La x_1)\ldots \hA_{n,\la}^{\out}(\La x_n)
\hat U(\la)\Psi_{\el}\non\\
&=& \hA_{1,\la}^{\out}(\eta_{+,r_1,\La})\ldots \hA_{n,\la}^{\out}(\eta_{+,r_n,\La})\hat U(\la)\Psi_{\el}. \label{poincare-computation}
\end{eqnarray}
We note  that $\eta_{+,r,\La}(x):=\eta_{+,r}(\La^{-1}x)$ belongs to the class of functions
defined in (\ref{test-functions-zero}), $\hA_{i,\la}\in \mfa_{C_*}$ and the timelike unit vectors entering into the construction of 
$\eta_{+,r_i,\La}$
and $\hA_{i,\la}$ coincide (cf. Definitions~(\ref{test-functions-zero}) and (\ref{local-C})).  It is easy to check that
\beqa  
\lim_{r_i\to\infty} A_{i,\la}^{\out}(\eta_{+,r_i,\La})\Om=
U(\la)A_i^{\out}\Om=U(\la)\Phi_i,
\eeqa
where $\Phi_i=A_{i}^{\out}\Om$. Thus, by taking the limits $r_i\to\infty$ on 
both sides of (\ref{poincare-computation}), we conclude the proof.
\hfill \qed\\
Let $\Ga(\mfh_{\pho})$ be the symmetric Fock space over $\mfh_{\pho}$ and we denote by $a^*(\,\cdot\,)$ and $a(\,\cdot\,)$
the corresponding creation and annihilation operators. Using Theorem~\ref{Compton-scattering} (a), (b), 
we  define the  \emph{outgoing wave operator} of Compton scattering 
\beqa
W^{\out}(\Ga(\mfh_{\pho})\otimes \hat\mfh_{\el})\to \hat\hil,
\eeqa
as the unique linear isometry, satisfying
\beqa
W^{\out}(a^*(\Phi_1)\ldots a^*(\Phi_n)\Om\otimes \Psi_{\el})=\hA_1^{\out+}\ldots \hA_n^{\out+}\Psi_{\el},
\eeqa
for $\Phi_i=A_i^{\out}\Om\in D_{\pho}$.  Setting $U_{\pho}(\la):=\Ga(U(\la)\res \mfh_{\pho})$ and
$\hat U_{\el}(\la):=\hat U(\la)\res \hat\mfh_{\el}$, we obtain from Theorem~\ref{Compton-scattering} (c)
\beqa
\hat U(\la)\circ W^{\out}=W^{\out}\circ ( U_{\pho}(\la)\otimes \hat U_{\el}(\la)),\quad \la\in \wt\mcP_+^{\uparrow},
\eeqa
which amounts to the Poincar\'e covariance of the wave operator.

\section{Haag-Kastler net of asymptotic photon fields}\label{last-section} \setcounter{equation}{0}

In this section we  construct  a Haag-Kastler net of asymptotic photon fields in a 
hypercone localized representation $\pi$ satisfying the properties specified at the beginning of Section~\ref{charged}.  
We also show that single-electron states induce vacuum representations of this net.

We start with the following technical lemma which summarizes and extends
the information about the domains of the asymptotic fields.
\bel\label{asymptotic-net-lemma-one} Let $\hA\in \hat\mfb^{C_*}$ be self-adjoint. Then:
\begin{enumerate}
\item[(a)] $D(\hat H)\subset D_{\max}(\hA)$ and $\hA^{\out}\res D(\hat H)$ is a symmetric operator uniquely determined by  $A^{\out}\Om$.
\item[(b)] $\|\hA^{\out}\Psi\|\leq c\|(1+\hat H)\Psi\|$, $\Psi\in D(\hat H)$.
\item[(c)] $|\lan \hat H\Psi, \hA^{\out}\Psi\ran-\lan \hA^{\out}\Psi, \hat H\Psi \ran|\leq c\|(1+H)^{1/2}\Psi\|^2$, $\Psi\in D(\hat H)$. 
\end{enumerate}
Moreover, $\i[\hat H, \hA^{\out}]$, defined as a quadratic form on $D(\hat H)\times D(\hat H)$, extends to
a symmetric operator $\i[\hat H, \hA^{\out}]^{\circ}$ on  $D(\hat H)$ in the sense explained in  Appendix~\ref{Weyl-rel}. This 
operator coincides  with $(\hA^{(1)})^{\out}\res D(\hat H)$, $\hA^{(1)}:=\i[\hat H, \hA]\in \hat\mfb^{C_*}$. Thus, it satisfies
properties (a),(b),(c) above.
\eel
\proof To prove (a), we note  that for any $\Psi\in D(\hat H)$ and $\eps>0$ there exists $\Phi_{\eps}\in D_{\hat H}$ such that 
$\|(1+\hat H)(\Psi-\Phi_{\eps})\|<\eps$. We write
\beqa
\bar{\hat A}_t\Psi=\bar{\hat A}_t\Phi_{\eps}+\bar{\hat A}_t(1+\hat H)^{-1}(1+\hat H)(\Psi-\Phi_{\eps})=\bar{\hat A}_t\Phi_{\eps}+O(\eps).
\label{D-H-extension}
\eeqa
Here $\|O(\eps)\|\leq c\eps$ uniformly in $t$ due to the energy bounds from Proposition~\ref{uniform-bound-prop}.
Now the existence of $\hat A^{\out}$ on $D(\hat H)$ follows from Proposition~\ref{charge-class-admissible-two} (a)
and the Cauchy criterion. It is  also clear from the above argument that $\hat A^{\out}\res D(\hat H)$ is uniquely 
determined by $\hat A^{\out}\res D_{\hat H}$. (That is, if $\hat A^{\out}\Phi=0$ for all $\Phi\in D_{\hat H}$ then 
$\hat A^{\out}\Psi=0$ for all $\Psi\in D(\hat H)$). 
Part (b) is a simple consequence of part (a) and Proposition~\ref{uniform-bound-prop}. 

To complete the proof of the lemma, we write for $\Psi\in D(\hat H)$
\beqa
\i\left(\lan \hat H\Psi, \hA^{\out}\Psi\ran-\lan \hA^{\out}\Psi, \hat H\Psi \ran\right)=
\lim_{t\to\infty}\lan \Psi, \bar{\hA}^{(1)}_t\Psi\ran=\lan \Psi, (\hA^{(1)})^{\out}\Psi\ran. \label{commutator-(1)}
\eeqa
Since the energy bounds give  $\|\bar{\hA}^{(1)}_t (1+\hat H)^{-1}\|\leq c$ and $\|(1+\hat H)^{-1}\bar{\hA}^{(1)}_t \|\leq c$,
uniformly in $t$, we obtain by interpolation  (cf. \cite[Appendix to IX.4]{RS2})
\beqa
\|(1+\hat H)^{-1/2}\bar{\hA}^{(1)}_t(1+\hat H)^{-1/2} \|\leq c, 
\eeqa
uniformly in $t$. This and the first equality in (\ref{commutator-(1)}) give part (c). The second
equality in (\ref{commutator-(1)}) ensures that $\i[\hat H, \hA^{\out}]^{\circ}$ is defined on $D(\hat H)$
and coincides on this domain with $(\hA^{(1)})^{\out}$. \hfill\qed\\
In the next lemma we collect the necessary information about the commutators of asymptotic fields.
\bel\label{asymptotic-net-lemma-two} Let $\hA_1, \hA_2\in \hat\mfb^{C_*}$ be self-adjoint. Then:
\begin{enumerate} 
\item[(a)]  $[\hA_1^{\out}, \hA_2^{\out}]=\lan \Om, [A_1^{\out}, A_2^{\out}]\Om \ran  1_{\hat\hil}$ as quadratic forms on $D(\hat H)\times D(\hat H)$.
\item[(b)] $\lan \Om, [A_1^{\out}, A_2^{\out}]\Om \ran=0$ if $A_1, A_2$ are localized in spacelike separated double cones.
\end{enumerate}
\eel
\proof We set $\eps>0$ and choose $\Psi_i\in D(\hat H)$ and $\Phi_{i,\eps}\in D_{\hat H}$, $i=1,2$, such that 
$\|(1+\hat H)(\Psi_i-\Phi_{i,\eps})\|<\eps$ as in the proof of Lemma~\ref{asymptotic-net-lemma-one}.
By taking the limit $t\to\infty$ in (\ref{D-H-extension}), we obtain
\beqa
\hat A_{i}^{\out}\Psi_{i'}=\hat A^{\out}_{i}\Phi_{i',\eps}+O(\eps),
\eeqa
where $\|O(\eps)\|\leq c\eps$. Thus, we get from Proposition~\ref{commutator-proposition-new}
\begin{align}
\lan \Psi_1,[\hA_1^{\out}, \hA_2^{\out}]\Psi_2\ran&=\lan \Phi_{1,\eps},[\hA_1^{\out}, \hA_2^{\out}]\Phi_{2,\eps}\ran+O(\eps)\non\\
&=\lan \Om, [A_1^{\out}, A_2^{\out}]\Om\, \ran \lan \Phi_{1,\eps},\Phi_{2,\eps}\ran+O(\eps)\non\\
&=\lan \Om, [A_1^{\out}, A_2^{\out}]\Om\, \ran \lan \Psi_{1},\Psi_{2}\ran+O(\eps),
\end{align}
where the rest term $|O(\eps)|\leq c\eps$ changes from line to line. Since $\eps$ was arbitrary,
this completes the proof of part (a). Part (b)  is a known consequence of the JLD representation (cf. \cite[p.160]{Bu77}).\hfill\qed\\
Making use of Lemmas~\ref{asymptotic-net-lemma-one}, \ref{asymptotic-net-lemma-two} and standard results
about self-adjoint extensions of unbounded operators collected in Appendix~\ref{Weyl-rel}, (where we set $N=1+\hat H$), 
we obtain the following proposition.
\bep\label{Weyl-proposition} Let $\hA_1, \hA_2\in \hat\mfb^{C_*}$ be self-adjoint.
Then $\hA_1^{\out}$ and $\hA_2^{\out}$ are essentially self-adjoint on $D(\hat H)$
and their self-adjoint extensions $\hA_1^{\out\ext}$, $\hA_2^{\out\ext}$ are essentially
self-adjoint on any core for $\hat H$. Moreover,
\beqa
\e^{\i (\hA_1+\hA_2)^{\out\ext}  }=\e^{\fr{1}{2}\lan \Om,[A^{\out}_1, A^{\out}_2]\Om\ran } \e^{\i\hA^{\out\ext}_1}\e^{\i\hA^{\out\ext}_2}.
\eeqa
\eep
Due to the previous results we are now in a positon to define the net of asymptotic photon fields. For any $\mco\in \mcK$
we, thus, introduce the von Neumann algebra
\beqa
\hat\mfa^{\out}(\mco):=\{\, \e^{\i\hA^{\out\ext}} \,|\, \hA\in \hat\mfb^{C_*}(\mco), \ \hA^*=\hA\,\}''.\label{asymptotic-algebra}
\eeqa
\bet $(\hat\mfa^{\out}, \hat U)$ is a Haag-Kastler net in the sense of Definition~\ref{HK}.
\eet
\proof Locality follows from Proposition~\ref{Weyl-proposition} and Lemma~\ref{asymptotic-net-lemma-two} (b). To show covariance
under Poincar\'e transformations, we use Lemma~\ref{Poincare-charged}, which gives on $D_{\hat H}$
\beqa
\hat U(\la)\hA^{\out}\hat U(\la)^*=\hA_{\la}^{\out}.
\eeqa 
Since $\hA_{\la}^{\out}$ is essentially self-adjoint on $D_{\hat H}$, which, moreover, is a core for $\hat H$,
all its self-adjoint extensions must coincide. In particular, we obtain
\beqa
\hat U(\la)\hA^{\out\ext}\hat U(\la)^*=\hA_{\la}^{\out\ext}. 
\label{ext-poincare-invariance}
\eeqa
Isotony and positivity of energy are obvious. \hfill\qed

Next, we proceed to a  discussion of representations of $(\hat\mfa^{\out}, \hat U)$ induced by vector states from $\hat\mfh_{\el}$.
\bel\label{expansion-lemma} Let $\hA\in \hat\mfb^{C_*}$ be self-adjoint and $\Psi_{\el}\in \hat\mfh_{\el}\cap D_{\hat H}$, $\|\Psi_{\el}\|=1$.
Then,
\beqa
\lan\Psi_{\el}, \e^{\i \hA^{\out\ext}} \Psi_{\el}\ran=
\e^{-\fr{1}{2}\| A^{\out}\Om\|^2}. \label{vacuum-relation}
\eeqa
\eel
\proof Consider the function $f(s):=\lan\Psi_{\el}, \e^{\i s \hA^{\out\ext}} \Psi_{\el}\ran$.
Since $\Psi_{\el}\in D_{\hat H}$ is contained in the domain of  $\hA^{\out\ext}$, we have by the Stone theorem
\beqa
(-\i)\pa_sf(s)=\lan\Psi_{\el}, \e^{\i s \hA^{\out\ext}} \hA^{\out\ext}\Psi_{\el}\ran=\lan\Psi_{\el}, \e^{\i s \hA^{\out\ext}} \hA^{\out}\Psi_{\el}\ran.
\eeqa
As $\hA^{\out}D_{\hat H}\subset D_{\hat H}$, we can iterate. This gives in particular
\beqa
(-\i)^n\pa_s^nf(s)|_{s=0}=\lan\Psi_{\el}, (\hA^{\out})^n\Psi_{\el}\ran.
\eeqa 
Now we use Proposition~\ref{creation-annihilation} to decompose $\hA^{\out}=\hA^{\out+}+\hA^{\out-}$ on $D_{\hat H}$ while keeping in mind that $\hA^{\out\pm}D_{\hat H}\subset D_{\hat H}$. Due to the canonical commutation relations (\ref{commutation-two}),
the fact that $\hA^{\out-}\Psi_{\el}=0$ (Corollary~\ref{main-ergodic-corollary}) and standard combinatorics we, moreover, have for even $n\geq 2$ 
\beqa
\lan\Psi_{\el}, (\hA^{\out})^n\Psi_{\el}\ran=(n-1)!! \lan \Om, (A^{\out})^2\Om\ran^{n/2}
\eeqa 
and zero for odd $n\geq 1$. Thus, we obtain
\beqa
\sum_{n=0}^{\infty } \fr{\i^n s^n}{n!}\lan\Psi_{\el}, (\hA^{\out})^n\Psi_{\el}\ran=\sum_{\ell=0}^{\infty}\fr{(-1)^{\ell}s^{2\ell} }{2^{\ell}\ell!} 
 \lan \Om, (A^{\out})^2\Om\ran^{\ell}=\e^{-\fr{1}{2}s^2\| A^{\out}\Om\|^2}, \label{convergent-sums}
\eeqa
where we set $\ell=n/2$. Since the sum on the left-hand side above is absolutely convergent for any $s\in \complex$, we conclude that $f$ extends to an entire
analytic function which  coincides with the function on the right-hand side of (\ref{convergent-sums}). \hfill\qed
\bet Let $\Psi_{\el}\in \hat\mfh_{\el}\cap D_{\hat H}$, $\|\Psi_{\el}\|=1$, $\om_{\el}(\,\cdot \,):=\lan \Psi_{\el}, \,\cdot \, \Psi_{\el}\ran$
be the corresponding state on $\hat\mfa^{\out}$ and $(\pi_{\el}, \hil_{\pi_\el}, \Om_{\pi_\el})$ its GNS representation. 
Then, $\pi_{\el}$ is a vacuum representation of $(\hat\mfa^{\out}, \hat U)$  in the sense of Definition~\ref{vacuum-def}. 
\eet
\begin{remark}  It is easy to see that the above theorem also  holds if $\pi$ is the original vacuum representation and $\Psi_{\el}$ is replaced with $\Om$.
This gives a different proof of a result from \cite{Bu77}.
\end{remark}
\begin{remark} If $U\res\mfh_{\pho}$ is an irreducible representation of
$\wt{\mathcal{P}}_+^{\uparrow}$ of
zero mass and some finite, integer spin, one could expect that $(\hat\mfa^{\out}_{\pi_{\el}}, \hat U_{\pi_{\el}})$ is just the usual massless free field of this spin on the Fock space. 
It turns out that  this is not true under our assumptions and counter-examples can be given using the following simple fact. 
Let $\mco_r$ be the standard double cone of radius $r$ centered at zero. Consider a Haag-Kastler net s.t.  $\mfa(\mco_r)=\complex I$ for $r<1$ and $\mfa(\mco_r)\neq \complex I$ for $r\geq 1$. Then also $\hat\mfa^{\out}(\mco_r)=\complex I$ for $r<1$.

\end{remark}
\proof First,  for $\hA$ as in Lemma~\ref{expansion-lemma}, we have
\beqa
& &\lan \Psi_{\el},\hat U(\la) \e^{\i \hA^{\out\ext}} \hat U(\la)^* \Psi_{\el}\ran\non\\
& &=\lan \Psi_{\el}, \e^{\i \hA_{\la}^{\out\ext}}\Psi_{\el}\ran=\e^{-\fr{1}{2}\| A_{\la}^{\out}\Om\|^2}=\e^{-\fr{1}{2}\| A^{\out}\Om\|^2}
=\lan \Psi_{\el}, \e^{\i \hA^{\out\ext}} \Psi_{\el}\ran. \label{state-invariance}
\eeqa 
Since by Proposition~\ref{Weyl-proposition} any element of $\hat\mfa^{\out}_{\loc}$ is a strong limit of finite linear combinations of  
operators of the form $\e^{\i \hA^{\out\ext}}$, we obtain that $\om_{\el}$ is invariant under Poincar\'e transformations.  Thus, by the GNS theorem \cite[Theorem 2.33]{Ar} 
we obtain a unique group of unitaries $U_{\pi_\el}$ acting on $\hil_{\pi_\el}$ such that 
\beqa
U_{\pi_\el}(\la)\Om_{\pi_\el}=\Om_{\pi_\el}, \quad U_{\pi_\el}(\la)\pi_{\el}(B)U_{\pi_\el}(\la)^*=\pi_{\el}(  \hat U(\la)B\hat U(\la)^*),
\eeqa
for all $B\in \mfa^{\out}$ and $\la\in \wt\mcP_+^{\uparrow}$. Weak (and therefore strong) continuity of $U_{\pi_\el}$ follows from the identity
\beqa
\lan \pi_{\el}(B_1)\Om_{\pi_\el}, U_{\pi_\el}(\la) \pi_{\el}(B_2)\Om_{\pi_\el}\ran=\lan \Psi_{\el}, B_1^* \hat U(\la) B_2 \hat U(\la)^*\Psi_{\el}\ran, \quad B_1, B_2\in \mfa^{\out}, \label{GNS-identity}
\eeqa 
and strong continuity of $\hat U(\la)$. Thus, $\pi_{\el}$ is a covariant representation. 

Positivity of energy of $\pi_{\el}$ easily follows from the above information,
using ideas from \cite[Theorem 2.2]{Dy08} and \cite[Theorem 4.5]{Ar}. More precisely, let $B\in  \hat\mfa^{\out}_{\loc}$ and $f\in S(\real^4)$ be such that $\supp\,  \wt f\cap \ov{V}_+=\emptyset$. Then, we obtain
from (\ref{GNS-identity}) that
\beqa
\| \wt f(H_{\el}, \bP_{\el}) \pi_{\el}(B)\Om_{\pi_\el}\|^2=\lan \Psi_{\el}, B(f)^*B(f)\Psi_{\el}\ran,
\eeqa
where $(H_{\el}, \bP_{\el})$ are the generators of $U_{\pi_\el}$. To show that the right-hand side above is zero,
we introduce compact sets $K\subset \real^3$, $\De_R=\{\, p\in \ov{V}_+\,|\, p^0\leq R\,\}$  and write
\beqa
& &\lan \Psi_{\el}, B(f)^*B(f)\Psi_{\el}\ran=
\lan E(\De_R)\Psi_{\el}, B(f)^*B(f)E(\De_R)\Psi_{\el}\ran+O(R^{-N})\non\\
& &\ph{4444}=\fr{1}{|K|}\int_{K} d^3x\,\lan E(\De_R)\Psi_{\el}, (B(f)^*B(f))(\bx) E(\De_R)\Psi_{\el}\ran+O(R^{-N}), 
\eeqa
where $|O(R^{-N})|\leq C_N/R^N$ for any $N\in\nat$. Here in the first step we used that $\Psi_{\el}\in D_{\hat H}$. 
In the second step we exploited the translation invariance of the functional on $\hat\mfa^{\out}$, which is induced by $E(\De_R)\Psi_{\el}\in \hat\mfh_{\el}\cap D_{\hat H}$. This invariance is
proven as in (\ref{state-invariance}).
By taking first the limit $K\nearrow \real^3$
and making use of Lemma~\ref{HA-lemma}, and then taking the limit $R\to\infty$, we conclude the proof of positivity of energy.

It remains to show the irreducibility of $\pi_{\el}$. As usually, it suffices to verify the clustering property, i.e.
\beqa
\lim_{|\bx|\to\infty}\om_{\el}(B_1B_2(\bx))=\om_{\el}(B_1)\om_{\el}(B_2), \quad B_1, B_2\in \hat\mfa^{\out}_{\loc}. \label{clustering-prop}
\eeqa  
In fact, given (\ref{clustering-prop}) and the fact that, by the Mean Ergodic Theorem, 
\beqa
E_{\mathrm{inv}}=\slim_{K\nearrow \real^3}\fr{1}{|K|}\int_{K}d^3x\, U_{\pi_\el}(\bx)
\eeqa
is a projection on invariant vectors of $U_{\pi_\el}\res \real^3$, we obtain that $E_{\mathrm{inv}}=|\Om_{\pi_\el}\ran\lan \Om_{\pi_\el}|$. Then,
\cite[Theorem~4.6]{Ar} gives irreducibility of $\pi_{\el}$. 

Let us verify (\ref{clustering-prop}) first for operators of the form $\e^{\i \hA_1^{\out\ext}}, \e^{\i \hA_2^{\out\ext}}$, where $\hA_1,\hA_2$
are as in Lemma~\ref{expansion-lemma}.  Taking the Weyl relations and (\ref{vacuum-relation}) into account, we obtain
\begin{align}
\om_{\el}(\e^{\i \hA_1^{\out\ext}} \e^{\i \hA_{2,\bx}^{\out\ext}} )
&=\e^{-\fr{1}{2}\lan \Om,[A^{\out}_1, A^{\out}_{2,\bx}]\Om\ran }\om_{\el}(\e^{\i (\hA_1+\hA_{2,\bx})^{\out\ext}  }  )\non\\
&=\e^{-\fr{1}{2}\lan \Om,[A^{\out}_1, A^{\out}_{2,\bx}]\Om\ran } \e^{   -\h \|(A_1+A_{2,\bx})^{\out}\Om\|^2  }.
\end{align}
It is well known that in a vacuum representation $\wlim_{|\bx|\to\infty} U(\bx)=|\Om\ran\lan \Om|$. (Observe that
$\lan B^*\Om, U(\bx) B^*\Om\ran=\lan \Om, [B, B^*(\bx)]\Om\ran\to 0$ for all $B$ as in Lemma~\ref{HA-lemma}).
Hence,
\beqa
\lim_{|\bx|\to\infty}\om_{\el}(\e^{\i \hA_1^{\out\ext}} \e^{\i \hA_{2,\bx}^{\out\ext}} )=\om_{\el}(\e^{\i \hA_1^{\out\ext}}) 
\om_{\el}(\e^{\i \hA_2^{\out\ext}}), \label{simple-clustering}
\eeqa
and this relation extends to finite linear combinations of operators of the form $\e^{\i \hA^{\out\ext}}$.

Now for any $B_1, B_2\in \hat\mfa^{\out}_{\loc}$  we can find, by the Kaplansky Density Theorem,  finite linear combinations 
$B_{1,\eps}$, $B_{2,\eps}$ such that $\|B_{1,\eps}\|\leq c$, $\|B_{2, \eps}\|\leq c$ uniformly in $\eps$ and
 $\|(B_1^*-B_{1,\eps}^*)\Psi_{\el}\|\leq \eps$, $\|(B_2-B_{2,\eps})\Psi_{\el}\|\leq \eps$. Making use of the translation invariance
of $\om_{\el}$ and relation~(\ref{simple-clustering}), we  write
\begin{align}
\om_{\el}(B_1B_2(\bx))&=\om_{\el}(B_{1,\eps}B_{2,\eps}(\bx))+O(\eps)\non\\
&=\om_{\el}(B_{1,\eps})\om_{\el}(B_{2,\eps})+O(\eps)+o_{\eps}(|\bx|^0)\non\\
&=\om_{\el}(B_{1})\om_{\el}(B_{2})+O(\eps)+o_{\eps}(|\bx|^0),
\end{align}
where $|O(\eps)|\leq c\eps$ uniformly in $|\bx|$ and $\lim_{|\bx|\to\infty}o_{\eps}(|\bx|^0)=0$.
By taking first the limit $|\bx|\to \infty$ and then $\eps\to 0$ we conclude the proof. \hfill\qed

\appendix

\section{Mean Ergodic Theorem and invariant vectors} \setcounter{equation}{0}

We pick $h$ as in (\ref{timeAverage}) and recall a variant of the abstract Mean Ergodic Theorem:  
\bet\label{MWLem}
Let $S$ be a self-adjoint operator on (a domain in) $\mathcal{H}$ and $F_S$ its spectral measure. Then,
\begin{equation}\label{MW1}
\slim_{t\to\infty}\int dt'\,h_t(t')\e^{\i t'S}=F_S(\{0\}).
\end{equation}
\eet
 Now we determine the projection $F_S(\{0\})$ on the subspace of invariant vectors of $t\mapsto \e^{\i tS}$ for the relevant operators $S$.
\bep\label{Ergodic-corollary} Let $(H,\bP)$ be the  energy-momentum operators of a  Haag-Kastler theory and
$E$ their joint spectral measure.
\begin{enumerate}
\item[(a)] Let $S_{\nu}:= H-\cos\,\nu |\bP|$ and $F_{S_{\nu}}$ be the spectral measure of $S_{\nu}$.  Then,
\begin{align} \label{zero-ergodic-one}
F_{S_{\nu}}(\{0\}) &=
\left\{ \begin{array}{ll} E(\pa\ov{V}_+) & \textrm{for  $\nu=0$}, \\ 
 E(\{0\}) & \textrm{for $\nu\in (0,\pi]$.} \\
\end{array} \right.
\end{align}
\item[(b)] Let $S_{\nu,\bla}:=H-|\bla|\cos\,\nu -\om_m(\mathbf{P}+\bla)$, where $\om_m(\bp)=\sqrt{\bp^2+m^2}$, and $F_{S_{\nu,\bla}}$ be the spectral measure of $S_{\nu,\bla}$. 
Then, for $\bla\neq 0$, 
\begin{align} \label{zero-ergodic-two}
F_{ S_{\nu,\bla} }(\{0\})  &=
\left\{ \begin{array}{ll} 0 & \textrm{for  $\nu\in [0,\pi)$ or $m>0$,} \\ 
 E(\{0\}) & \textrm{for $\nu=\pi$ and $m=0$.} 
\end{array} \right.
\end{align}
\end{enumerate}
\eep
\proof 
(a) For $\Psi_{0}\in\Ran F_{S_0}(\{0\})$, we have $(H-|\mathbf{P}|)\Psi_0=0$, hence $\Psi_0\in\Ran E(\partial \overline{V}_+)$. This gives the first part of (\ref{zero-ergodic-one}). 
To check the second part, we note that for $\nu\in (0,\pi]$ the set
\beqa
\De_{\nu}:=\{\, (p^0, \bp)\,|\, p^0=\cos \nu\,|\bp| \,\}
\eeqa
intersects with $\ov{V}_+$ only at $\{0\}$. 

(b) First, we note that the set
\beqa
\De_{\nu,\bxi}:=\{\, (p^0, \bp)\,|\, p^0=|\bxi|\cos \nu+\om_m(\bp+\bxi)   \,\}
\eeqa
describes a mass hyperboloid shifted by a spacelike or lightlike vector $(|\bxi|\cos\nu,-\bxi)$.
Thus $\De_{\nu,\bxi}$ contains zero only if $m=0$ and $\nu=\pi$. Hence, it suffices to show that
the relation
\beqa
(H-\om_m(\bP-\bxi))\Psi=|\bxi|\cos\nu\Psi, \label{starting-point}
\eeqa
where $\Psi=E(\De)\Psi$, $\De$ compact, can only hold for $\Psi\in E(\{0\})\hil$.

To this end, we generalize an argument from the Appendix of \cite{Bu75}: From (\ref{starting-point}) we obtain
\begin{align}
(H^2-|\bP-\bxi|^2-m^2)\Psi&=|\bxi|\cos\nu(H+\om_m(\bP-\bxi))  \Psi\non\\
&=|\bxi|\cos\nu(2H-|\bxi|\cos\nu)  \Psi.
\end{align}
Setting $M^2:=H^2-\bP^2$, we get
\beqa
M^2\Psi
=(2H|\bxi|\cos\nu-2\bP\bxi+|\bxi|^2 \sin^2\nu+m^2)  \Psi. \label{before-transformation}
\eeqa
Now we want to apply a Lorentz transformation $U(\wt \La)$ to both sides of the above equation. 
We recall from Subsection~\ref{Lorentz-subsection} that
\beqa
& &U(\wt \La)HU(\wt \La)^* =\lan \bv_{\La^{-1}} \ran H + 
\bv_{\La^{-1}} \bP,\label{H-transform-one}\\
& &U(\wt \La)\bP U(\wt \La)^* =- \bv_{\La} H+[\La^{-1}] \bP. \label{P-transform-one}
\eeqa
Choosing $\La=\La_{\eta}$ to be a boost with rapidity $\eta$ in some direction $\bn$ orthogonal to $\bxi$, we get that 
$\bv_{\La_{\eta}}$ is orthogonal to $\bxi$ and  $[\La^{-1}_{\eta}]^T\bxi=\bxi$. Therefore,
\beqa
M^2\Psi_{\La_{\eta}}=\big( 2(\lan \bv_{\La^{-1}_{\eta}}\ran H+\bv_{\La^{-1}_{\eta}}  \bP) |\bxi|\cos\,\nu-2\bP \bxi+|\bxi|^2\sin^2\nu +m^2\big)  \Psi_{\La_{\eta}}.
\eeqa
Taking the scalar product with $\Psi$ and making use of (\ref{before-transformation}) we get for $\bxi\neq 0$ and $\cos\,\nu\neq 0$
\beqa
\lan\Psi,H \Psi_{\La_{\eta}}\ran(1-\lan\bv_{\La^{-1}_{\eta}}\ran )=\lan \Psi, (\bv_{\La^{-1}_{\eta}} \bP) \Psi_{\La_{\eta}}\ran. \label{orthogonal-equation}
\eeqa
We note that the term on the left-hand side above is of order $\eta^2$
while the term on the right-hand side is of order $\eta$. Thus, dividing both sides of the equation by $\eta$
and taking the limit $\eta\to 0$, we obtain
\beqa
\lan \Psi, \bn \bP \Psi\ran=0.
\eeqa
Since the above equation holds also for $\Psi$ replaced with $E(\De_{\pm})\Psi$, 
where $\De_{\pm}$ are chosen so that $\pm E(\De_{\pm})\bn \bP E(\De_{\pm})\geq 0$,
we conclude that $(\bn \bP)\Psi=0$. Substituting this to (\ref{orthogonal-equation}) we infer
that $\lan \Psi, H\Psi\ran=0$ and, therefore, $\Psi\in E(\{0\})\hil$.   

In the case of $\cos\nu=0$ and $\bxi\neq 0$ we choose $\La_{\eta}$ to be the boost with rapidity $\eta$  in the direction of $\bxi$. 
Then, an analogous reasoning as above gives
\beqa
\lan \Psi, \big((1- [\La_{\eta}^{-1}]) \bP\big) \bxi\,\Psi_{\La_{\eta}}\ran=-\lan \Psi, H\Psi_{\La_{\eta}}\ran \bv_{\La_{\eta}^{-1}}\bxi.
\eeqa
Since $(1- [\La_{\eta}^{-1}])^{T}\bxi$ is of order $\eta^2$, we obtain $\lan \Psi, H\Psi\ran=0$, which concludes the proof.\hfill\qed

\section{Admissible propagation observables} \label{admissible-appendix} \setcounter{equation}{0}

\bed\label{admissible-prop-obs} Let $[1,\infty)\ni t\mapsto A_t\in B(\hil)$ be a propagation observable,
 $H$ a self-adjoint operator on a domain $D(H)$ in $\hil$, and $D,D^*\subset \hil$ some dense domains. 
We say that $A$ is admissible if:
 \begin{itemize}
\item[(a)]  For any $\Psi\in D^{(*)}$
the limit $\lim_{t\to\infty}A_t^{(*)}\Psi$ exists.
\item[(b)] $\sup_{t\in\real_+}\|A_t^{(*)}(1+H)^{-1}\|<\infty$.
\item[(c)] Set $A_t(s):=\e^{\i sH}A_t\e^{-\i sH}$. All the derivatives $A_t^{(n)}=\pa^n_s A_t(s)|_{s=0}$ exist in norm and satisfy (a), (b).
\end{itemize}
Here $(*)$ means that the statement holds  with and without all $*$ symbols (correlated).
\eed

As shown in the next two propositions, limits of admissible propagation observables exist as closable operators on
the following dense domain 
\beqa
D_H:=\bigcap_{n\geq 1} D(H^n).
\eeqa
Moreover, $D_H$ is an invariant domain of these limits.  
\bep\label{admissible-proposition-one} Let $A$ be  an admissible propagation observable. Then:
\begin{enumerate}
\item[(a)] For any $\Psi\in D_H$ the limit $\lim_{t\to\infty}A_t\Psi$ exists and defines a closable operator $A^{\out}$ on $D_H$.
This operator is uniquely specified by its values on $D$.
\item[(b)] $A^{\out} D_H\subset D_H$.
\end{enumerate}
\eep
\proof Exploiting part (c) of Definition~\ref{admissible-prop-obs}, we write
\beqa
A_t\Psi=(1+H)^{-1}(-\i)A^{(1)}_t\Psi+(1+H)^{-1}A_t(1+H)\Psi. \label{limit-A-t}
\eeqa
To prove (a), we 
use Definition~\ref{admissible-prop-obs} (b), (c) to approximate vectors $\Psi,(1+H)\Psi$ by elements
of $D$ uniformly in $t$. By part (a) of Definition~\ref{admissible-prop-obs}, $A_t, A^{(1)}_t$ converge on $D$ which
gives the existence of $A^{\out}$ as an operator on $D_H$. Since the above reasoning applies also to $A_t^*$, the operator
$A^{\out}$ is closable. To show that it is uniquely determined by its values on $D$, consider admissible propagation observables $A_1$ and $A_2$ such that $\lim_{t\to\infty}A_{1,t}\Phi=\lim_{t\to\infty}A_{2,t}\Phi$ for $\Phi\in D$. Then, it is clear from the above discussion that $A_1^{\out}=A_2^{\out}$ as operators on $D_H$. This completes the proof of (a). 

To prove (b), we make use of a standard commutator formula (see e.g. \cite{FGS01})
\begin{align}
[(1+H)^\ell,A_t]&=\sum_{k=1}^{\ell}\begin{pmatrix}
\ell \\k
\end{pmatrix} \text{ad}_{H}^k(A_t)(1+H)^{\ell-k},\label{commutator-first}\\
\text{ad}_H^0(A_t):=A_t,&\qquad\text{ad}_H^k(A_t):=[H,\text{ad}_H^{k-1}(A_t)],
\end{align}
which holds as an equality of quadratic forms on $D_H\times D_H$. Exploiting part (c) of Definition~\ref{admissible-prop-obs},
which ensures that $\text{ad}_{H}^k(A_t)=(-\i)^kA^{(k)}_t$  are bounded operators, we obtain for any $\Psi\in D_H$
\beqa
A_t\Psi=(1+H)^{-\ell}\bigg(\sum_{k=0}^{\ell}\begin{pmatrix}
\ell \\k
\end{pmatrix} (-\i)^kA^{(k)}_t   (1+H)^{\ell-k}\bigg)\Psi, \label{commutator-formula}
\eeqa
where we set by convention $A_t^{(0)}=A_t$. Taking now the limit $t\to\infty$ on both sides of (\ref{commutator-formula}),
we obtain (b). \hfill\qed
\bep\label{admissible-proposition-two} Let $A_i$, $i=1,\ldots, n$, be admissible propagation observables. 
Then, for any $\Psi\in D_H$,
\beqa
A_1^{\out}\ldots A_n^{\out}\Psi=\lim_{t\to\infty}A_{1,t}\ldots A_{n,t}\Psi. 
\eeqa 
\eep
\proof For $n=1$ the statement follows from Proposition~\ref{admissible-proposition-one}.
We suppose now it holds for some $n>1$ and prove it for $n+1$. Similarly as in (\ref{limit-A-t}), we write
for any $\Psi\in D_H$
\begin{align}
A_{1,t}\ldots A_{n+1,t}\Psi&=A_{1,t}(1+H)^{-1}(-\i) \sum_{\ell=2}^{n+1}A_{2,t}\ldots A^{(1)}_{\ell,t} \ldots A_{n+1,t}\Psi\non\\
&\ph{44}+A_{1,t}(1+H)^{-1}A_{2,t}\ldots A_{n+1,t}(1+H)\Psi.
\end{align}
By the induction hypothesis and Proposition~\ref{admissible-proposition-one} the above expression converges 
strongly as $t\to\infty$. Next, we pick $\Phi\in D_H$ and write
\begin{align}
\lan \Phi, A_{1,t}\ldots A_{n+1,t}\Psi\ran&=\lan \Phi, A_{1,t}A_{2}^{\out}\ldots A_{n+1}^{\out}\Psi\ran+o(t^{0})\non\\
&=\lan \Phi, A_{1}^{\out}A_{2}^{\out}\ldots A_{n+1}^{\out}\Psi\ran+o(t^{0}),
\end{align}
where in the first step we used the induction hypothesis, in the second step Proposition~\ref{admissible-proposition-one} 
and $o(t^0)$ denotes terms which tend to zero as $t\to\infty$. This concludes the proof.\hfill \qed
\section{Geometric argument}\label{geometric}
\setcounter{equation}{0}
\nin\textbf{Proof of Lemma~\ref{geometric-one}:} First, we note that
\beqa
\bigcup_{t\geq 1} O^{\La}_t=
\bigcup_{t\geq 1} \bigcup\limits_{\tau\in t+t^{\beps}\supp h}\left\{\La\mathcal{O}+\tau(1, g_{\La}(\Theta))\right\}
\subset  \bigcup_{\tau \in \real_+}  \left\{\La\mathcal{O}+\tau(1, g_{\La}(\Theta))\right\}=:\mcU^{\La}. \label{mcu-def}
\eeqa
On the other hand,
\beqa
\bigcup_{t\geq 1} \La O_t&=&\bigcup_{t\geq 1} \bigcup\limits_{\tau\in t+t^{\beps}\supp h}\left\{\La\mathcal{O}+\tau\La(1,\Theta)\right\}
\subset \bigcup_{\tau \in \real_+}  \left\{\La\mathcal{O}+\tau\La(1,\Theta)\right\}\non\\
&\subset &\bigcup_{\tau \in \real_+}\left\{\La\mathcal{O}+
(\lan \bv_{\La}\ran+\bv_{\La} \Theta) \tau (1,  g_{\La}(\Theta) )\right\}=\mcU^{\La}.
\eeqa
Here in the third step we made use of (\ref{g-Lambda-zero}) and (\ref{g-Lambda}) and in the last step 
of the fact that the prefactors $(\lan \bv_{\La}\ran+\bv_{\La} \Theta)$ are strictly 
positive and, thus, they just reparametrize~$\tau$.

Let us first disregard the Lorentz transformations, i.e. show that for any double cone $\mco\in \mcK$ 
and open $\Theta\subset S^2$ with $\ov{\Theta}\subsetneq S^2$, 
there is a 
future lightcone $V$ and a hypercone $\mcC\subset \mcF_V$ such that
the corresponding set  $\mcU^{\La=I}$ is in $\mcC^{\cc}$. The extension of the statement to $\mcU^{\La}$, where $\La$ is in some neighbourhood of unity $N$ in $\mathcal L_+^{\uparrow}$, will be discussed in the last part of the proof.
 
First, we fix a future lightcone $V$ so that $\ov{\mco}\subset V$ and 
choose a coordinate frame in which the origin is at the apex of $V$. Next, use the fact that there is an $\bell_0\in S^2$ and an $1\geq \eps_0>0$ such that the spherical cap
\beqa
\Theta_{\eps}:=\{ \bell\in S^2\,|\, 1-\eps\leq \bell\bell_0\leq 1\,\}
\eeqa
is contained in $S^2\backslash \ov{\Theta}$ for all $0<\eps\leq \eps_0$.
Let, moreover, $\msf{K}_{\eps}$ be a cone in the unit ball $\msf{B}$
with apex at $\bu_{\eps}:=(1-\eps)\bell_0$  and the opening angle determined by $\Theta_{\eps}$. More precisely,
\begin{equation}
\msf{K}_{\eps}:=\left\{\bu\in \mathsf{B}\,|\, \bu=\bu_{\eps}+s\,(\bell-\bu_{\eps} ),\ 0\leq s<1, \ \bell\in \Theta_{\eps} \right\}.
\end{equation}
Using the  Beltrami-Klein map $\bv: \mathsf{H}_{\bar\tau}\to \mathsf{B}$ given by $\bv(a)=\ba/a^0$, the
corresponding hyperbolic cone $\msf{C}( \msf{K}_{\eps})\subset \msf{H}_{\bar\tau}$ is given by
\beqa
\msf{C}( \msf{K}_{\eps})=\left\{\, \bar\tau\fr{(1,\bu)}{\sqrt{1-\bu^2}} \in \msf{H}_{\bar\tau} \,\bigg|\, \bu=\bu_{\eps}+s\,(\bell-\bu_{\eps} ),\ 0\leq s<1, \ \bell\in \Theta_{\eps} \right\}.
\eeqa
We note that as $\eps\to 0$, the apex of $\msf{C}( \msf{K}_{\eps})$ tends to lightlike infinity in the direction of $\bell_0$ and the opening
angle tends to zero.
In fact, for all $0\leq s<1$ and $\bell\in \Theta_{\eps}$ we have
\beqa
\bu_{\eps}(s, \bell):=\bu_{\eps}+s\,(\bell-\bu_{\eps} )=\bell_0(1-\eps(1-s))+s(\bell-\bell_0). \label{first-u}
\eeqa 
Noting that $(\bell-\bell_0)^2=2(1-\bell\bell_0)\leq 2\eps$ and 
setting $\bh_{\eps}(s, \bell):=-\eps^{\h}\bell_0(1-s)+s\eps^{-\h}(\bell-\bell_0)$, we have
\beqa
& &\bu_{\eps}(s, \bell)=\bell_0+\eps^{\h} \bh_{\eps}(s, \bell), \\
& &|\bh_{\eps}(s, \bell)|\leq 3.   \label{h-conditions}
\eeqa
Now a simple computation using (\ref{first-u}) gives
\beqa
1-\bu_{\eps}(s, \bell)^2
=\eps(1-s) \big\{  2-\eps(1-s)   +2s(1-\eps)(1-\bell\bell_0)\eps^{-1}\big\}.
\eeqa
It is easy to see that $1\leq \{\ldots\}\leq 4$ and, therefore, we can find a function  $(s,\bell)\mapsto g_{\eps}(s, \bell)$ such that
 $\fr{\bar\tau}{2}\leq g_{\eps}(s, \bell)\leq \bar\tau$ and
\beqa
\bar\tau\fr{1}{\sqrt{1-\bu_{\eps}(s, \bell)^2}}= \fr{g_{\eps}(s, \bell)}{\sqrt{\eps(1-s)}   }.
\eeqa
Thus, skipping the arguments of $g, \bh$ and setting
 $M:=\eps^{-\h}$, $S:=g(1-s)^{-\h}$, we have
\beqa
\bar\tau\fr{(1,\bu_{\eps}(s, \bell))}{\sqrt{1-\bu_{\eps}(s, \bell)^2}} =MS(1,\bell_0)+S(0, \bh), \label{hyperboloid-tau-definition}
\eeqa
where $M$ takes values in $[\eps_0^{-1/2}, \infty)$ and $S$ in $[\fr{\bar\tau}{2}, \infty)$.

Let us now show that there is a $c>0$ such that for sufficiently large $M$
\beqa
(MS(1,\bell_0)+S(0, \bh)-x)^2<-c, \label{geometric-claim}
\eeqa
for all $x\in \mco$, $S\in [\fr{\bar\tau}{2}, \infty)$ and $\bh$ within the above restrictions.
Since $\ov\mco\subset V$, there are constants $c_{\mco},  c'_{\mco}$ such that
\beqa
0<c_{\mco} \leq (x^0\pm|\bx|)\leq  c'_{\mco},
\eeqa
uniformly in $x\in \mco$. Moreover, due to (\ref{hyperboloid-tau-definition}) we have $(MS(1,\bell_0)+S(0, \bh))^2=\bar\tau^2$. 
Hence,
\begin{align}
(MS(1,\bell_0)+S(0, \bh)-x)^2&=\bar\tau^2-2MS(x^0-\bx\bell_0)-2S(0, \bh)x+x^2\non\\
&\leq -2MSc_{\mco}+ 6Sc_{\mco}'+ (c'_{\mco})^2+\bar\tau^2,
\end{align}
which proves (\ref{geometric-claim}).

Next, let us show that there is a $c'>0$ such that for sufficiently large $M$
\beqa
(MS(1,\bell_0)+S(0, \bh)-x-\tau(1, \bell'))^2<-c',
\eeqa
for all $\tau\in \real_+$, $\bell'\in \Theta$, $x\in \mco$, $S\in [\fr{\bar\tau}{2}, \infty)$ and $\bh$ within the above restrictions.
In view of (\ref{geometric-claim}), it suffices to note the estimate
\begin{align}
\big(MS(1,\bell_0)+S(0, \bh)-x\big)(1, \bell')&= S( M(1-\bell_0\bell')-\bh\bell')-x(1,\bell')\non\\
&\geq (\bar{\tau}/2)\big(M\eps_0-3\big)-c_{\mco}'.
\end{align}
Thus, we have proven that $\mcU^{\La=I}\subset \msf{C}( \msf{K}_{\eps})^{\cc}=\mcC(\msf{K}_{\eps})^{\cc}$ for $\eps$ sufficiently small,
depending on $\mco$ and $\Theta$.

Finally, let us choose a double cone $\mco_0$, satisfying $\ov{\mco}_0\subset \mco$, and an open set $\Theta_0\subset S^2$, fulfilling $\ov{\Theta}_0\subset \Theta$. ($\mco_0$ and $\Theta_0$ are still arbitrary, within the restrictions of the lemma, since $\mco$ and $\Theta$ were arbitrary).
Then, there is clearly a neighbourhood of unity $N$ in the Lorentz group such that $\La\mco_0\subset \mco$ and
$g_{\La}(\Theta_0)\subset\Theta$ for all $\La\in N$ (cf. (\ref{Lorentz-continuity}) for the latter condition). Therefore, by the
first part of the proof,
\beqa
\mcU^{\La}_{0}\subset \mcC^{\cc}, \quad \La\in N,
\eeqa  
where $\mcU_{0}^{\La}$ is defined as in (\ref{mcu-def}) using $\mco_0$ and $\Theta_0$. \hfill \qed

\section{Integrating Heisenberg commutation relations to Weyl relations} \label{Weyl-rel}

We state below two known results which were used in Section~\ref{last-section}. The first one is the Nelson commutator 
theorem \cite[Theorem X.37]{RS2},\cite[Theorem 0']{Fr77}.
\bet\label{Nelson} Let $N$ be a self-adjoint operator on $D(N)$ with $N\geq 1$. 
Let $A$ be a symmetric operator on $\hil$ with domain $D(A)$ which contains $D(N)$.  
Suppose that
\beqa
\|A\Psi\|\leq c\|N\Psi\| \quad \mathrm{and} \quad  |\lan A\Psi, N\Psi\ran-\lan N\Psi, A\Psi\ran|\leq d\|N^{1/2}\Psi\|^2 \label{Nelson-bounds}
\eeqa
for all $\Psi\in D(N)$. Then, $A$ is essentially self-adjoint on $D(N)$ and its unique self-adjoint extension $A^\ext$ is essentially self-adjoint
on any core for $N$. 
\eet
Before we state the second result we need some preparations. Let $A,N$ be as in Theorem~\ref{Nelson}. We, then, define
\beqa
\dot A=\i[N,A]
\eeqa
as a quadratic form on $D(N)\times D(N)$.  The associated operator  $\dot A^{\circ}$ is given by
\begin{align}
D(\dot A^{\circ})&=\{\Psi\in D(N)\,|\, \exists\, c_{\Psi} \textrm{ s.t. } |\lan \,\Phi, \dot A \Psi\ran|\leq c_{\Psi}\|\Phi\|\textrm{ for all } \Phi\in D(N) \,\},\,\,\,\label{domain-def} \\
\dot A^{\circ}\Psi&=\dot{A}\Psi, \quad \Psi\in D(\dot A^{\circ}),
\end{align}
where the vector $\dot{A}\Psi$ corresponds via the Riesz theorem to the bounded functional appearing in (\ref{domain-def}).
It is easy to see that $\dot A^{\circ}$ is a symmetric operator on $D(\dot A^{\circ})$. However, it is not guaranteed that $D(\dot A^{\circ})$ is dense.

Now we are in a position to state a result about integration of canonical commutation relations from \cite[Theorem $\mathrm{1_M}$]{Fr77}.
(Although separability of $\hil$ is assumed in \cite{Fr77}, this property is not used in the proof of the following result).
\bet\label{Froehlich-Weyl} Let $N$ be a self-adjoint operator with $N\geq 1$. 
Let $A_1$, $A_2$  be symmetric operators with domains $D(A_1)$ and $D(A_2)$, containing $D(N)$, and  such that
\beqa
C:=\i[A_1, A_2], \label{C-def}
\eeqa
defined as a quadratic form on $D(N)\times D(N)$, is a multiple of the identity. Assume moreover that 
$D(\dot A^{\circ}_1)\supset D(N)$ and $A_1, A_2, \dot A^{\circ}_1$ satisfy (\ref{Nelson-bounds}). 
Then, the self-adjoint extensions $A_1^\ext$, $A_2^{\ext}$, given by Theorem~\ref{Nelson}, satisfy
\beqa
\e^{\i t A_1^{\ext}} \e^{\i s A_2^{\ext}} \e^{-\i t A_1^{\ext}}=\e^{\i s A_{2}^{\ext}}\e^{\i st C}, \quad s,t\in \real. \label{Froehlich-Weyl-rel}
\eeqa
\eet
  From Theorem~\ref{Froehlich-Weyl} we easily get the usual form of the Weyl relations appearing in 
 Proposition~\ref{Weyl-proposition}: 
\bec Let $N, A_1, A_2$ be as in Theorem~\ref{Froehlich-Weyl}. Then $A_1+A_2$, defined 
as a symmetric operator on $D(A_1)\cap D(A_2)$ is essentially self-adjoint on $D(N)$ and
its self-adjoint extension $(A_1+A_2)^\ext$ is essentially self-adjoint on any core for $N$. 
Moreover,
\beqa
\e^{\i t(A_1+A_2)^\ext}=\e^{-\fr{\i}{2} t^2 C}\e^{\i t A_1^{\ext}} \e^{\i t A_2^{\ext}}, \quad t\in \real,
\eeqa
with $C$ defined by (\ref{C-def}).
\eec
\proof To justify the first statement, we note that $A_1+A_2$ satisfies the assumptions of Theorem~\ref{Nelson}.
Now we define
\beqa
V(t):=\e^{-\fr{\i}{2} t^2 C}\e^{\i t A_1^{\ext}} \e^{\i t A_2^{\ext}}.
\eeqa
Clearly, $V(0)=1$ and, making use of (\ref{Froehlich-Weyl-rel}, we get 
\beqa
V(t)V(s)&=&\e^{-\fr{\i}{2} (t^2+s^2) C}\e^{\i t A_1^{\ext}} \e^{\i t A_2^{\ext}} \e^{\i s A_1^{\ext}} \e^{\i s A_2^{\ext}}\non\\
&=&\e^{-\fr{\i}{2} (t+s)^2 C}\e^{\i t A_1^{\ext}}  \e^{\i s A_1^{\ext}} \e^{\i t A_2^{\ext}}\e^{\i s A_2^{\ext}}=V(t+s).
\eeqa
Thus $V$ is a one-parameter group of unitaries, whose weak (and therefore strong) continuity is obvious. By the Stone
theorem $V(t)=\e^{\i t Q}$, for a self-adjoint operator $Q$ given by
\beqa
& &D(Q):=\{\, \Psi\in \hil\,|\, \lim_{\tau\to 0}\fr{V(\tau)-1}{\tau}\Psi \textrm{ exists} \,\}, \\
& &Q\Psi=\lim_{\tau\to 0}\fr{1}{\i}\fr{V(\tau)-1}{\tau}\Psi \textrm{ for } \Psi\in D(Q),
\eeqa 
cf. \cite[Theorems VIII.7, VIII.8]{RS1}. From the equality
\beqa
(V(\tau)-1)=(\e^{-\fr{\i}{2} \tau^2 C} -1)\e^{\i \tau A_1^{\ext}} \e^{\i \tau A_2^{\ext}}+\e^{\i \tau A_1^{\ext}} (\e^{\i \tau A_2^{\ext}}-1)
+(\e^{\i \tau A_1^{\ext}}-1)
\eeqa
and the Stone theorem we immediately conclude that $D(N)\subset D(Q)$, (since $D(N)\subset D(A_1^\ext)\cap D(A_2^\ext)$), 
and that
\beqa
Q \res D(N)=(A_1+A_2)\res D(N).
\eeqa
Thus, $Q$ is a self-adjoint extension of $(A_1+A_2)\res D(N)$ and by the first part of the theorem we
obtain $Q=(A_1+A_2)^\ext$. \hfill\qed

\section{Conventions}\label{conventions} \setcounter{equation}{0}

\begin{enumerate}
\item  $\wt g(p^0)=(2\pi)^{-1/2}\int \e^{\i p^0x^0}g(x^0)dx^0$ for $g\in L^1(\real)$.
\item  $\wt g(\bp)=(2\pi)^{-3/2}\int \e^{-\i \bp\bx}g(\bx)d^3x$ for $g\in L^1(\real^3)$.
\item  $\wt g(p)=(2\pi)^{-2}\int \e^{\i(p^0x^0-\bp\bx)}g(x)d^4x$ for $g\in L^1(\real^4)$.
\item $\wt T(p)=(2\pi)^{-2}\int \e^{-\i(p^0x^0-\bp\bx)}T(x)d^4x$ for $T\in S'(\real^4)$.
\item  $(f\ast g)(\bx)=\int \,f(\bx-\by)g(\by) d^3y$ for $f,g\in L^1(\real^3)$.
\item  $(f\ast_3 g)(x)=\int \,f(x^0,\bx-\by)g(\by) d^3y$ for $f\in L^1(\real^4),\,g\in L^1(\real^3)$.
\end{enumerate}



\begin{thebibliography}{LNT2}
	
\bibitem[Ar]{Ar} H. Araki. \emph{Mathematical theory of quantum fields}. Oxford University Press, 1999.	
	
\bibitem[Ar82]{Ar82} W. Arveson.
	\emph{The harmonic analysis of automorphism groups}. In
	Operator algebras and applications, Part I (Kingston, Ont., 1980),
	Proc. Sympos. Pure Math., 38,  Amer. Math. Soc., Providence, R.I.,1982.D., pp. 199--269.
	
	\bibitem[BDN14]{BDN14} S. Bachmann, W. Dybalski and P. Naaijkens. \emph{Lieb-Robinson bounds, Arveson spectrum and Haag-Ruelle scattering theory for gapped quantum spin systems.} To appear in Ann. Henri Poincar\'e. Preprint arXiv:1412.2970.
	
     \bibitem[Bu77]{Bu77} D. Buchholz. \emph{Collision theory for massless bosons}.
	Commun. Math. Phys. \bf 52\rm, (1977) 147--173.
	
	\bibitem[Bu75]{Bu75} D. Buchholz. \emph{Collision theory for massless fermions}. Commun. Math. Phys. \bf 42\rm, (1975) 269--279.
	
	
	\bibitem[Bu82]{Bu82} D. Buchholz.
	\emph{The physical state space of quantum electrodynamics}.
	Commun. Math. Phys. \bf 85\rm, (1982) 49--71. 
	
	\bibitem[Bu86]{Bu86} D. Buchholz.
	\emph{Gauss' law and the infraparticle problem}.
	Phys. Lett. B \bf 174\rm,  (1986) 331--334.
	
	\bibitem[Bu90]{Bu90} D. Buchholz.
	\emph{Harmonic analysis of local operators}.
	Commun. Math. Phys. \bf 129\rm, (1990) 631--641.
	
	\bibitem[BPS91]{BPS91} D. Buchholz, M. Porrmann and U. Stein.
	\emph{Dirac versus Wigner: Towards a universal particle
		concept in quantum field theory}.
	Phys. Lett. B  \bf 267\rm, (1991) 377--381.
	
	\bibitem[BR14]{BR14} D. Buchholz and J.E. Roberts. \emph{New light on infrared problems: sectors, statistics, symmetries and spectrum.}
	Commun. Math. Phys \bf 330\rm, (2014) 935--972.
	
	
	\bibitem[CFP07]{CFP07}  T. Chen, J. Fr\"ohlich and A. Pizzo, \emph{Infraparticle scattering states in non-relativistic QED: 
		I. The Bloch-Nordsieck paradigm.} Commun. Math. Phys. \bf 294\rm, (2010) 761--825.
	
	
	
	\bibitem[DH15]{DH15} P. Duch and A. Herdegen. \emph{Massless asymptotic fields and Haag-Ruelle theory.}
	Lett. Math. Phys. \textbf{105},  (2015) 245--277.
	
\bibitem[Dy08]{Dy08} W. Dybalski. \emph{A sharpened nuclearity condition and the uniqueness of the vacuum in QFT}. Commun. Math. Phys. \textbf{283}, (2008)
523--542.	
	

\bibitem[DG14]{DG14} W. Dybalski and C. G\'erard. \emph{A criterion for asymptotic completeness in local relativistic QFT.} Commun. Math. Phys. \bf 332, \rm (2014) 1167--1202. 
	
	
	
\bibitem[Fr77]{Fr77} J. Fr\"ohlich. \emph{Application of commutator theorems to the integration of representations of Lie algebras and commutator relations.}
Commun. Math. Phys. \bf 54, \rm (1977) 135--150.  


	\bibitem[FGS01]{FGS01} J. Fr{\"o}hlich, M Griesemer and  B. Schlein. \emph{Asymptotic electromagnetic fields in models of quantum-mechanical matter interacting with the quantized radiation field.} Advances in Mathematics \bf 164\rm, (2001) 349--398.
	
	


\bibitem[He14]{He14.0} A. Herdegen. \emph{Infraparticle problem, asymptotic fields and Haag-Ruelle theory}. Ann. Henri Poincar\'e \textbf{15}, 
	(2014) 345--367.
	
	\bibitem[He14.1]{He14} A. Herdegen. \emph{On energy momentum transfer of quantum fields}. Lett. Math. Phys. \textbf{104}, (2014) 1263–-1280.
	

\bibitem[Kr82]{Kr82} K. Kraus.
 \emph{Aspects of the infrared problem in quantum electrodynamics.} Found. Phys. \textbf{13}, (1983) 701--713.


\bibitem[KPR77]{KPR77} K. Kraus, L. Polley and G. Reents. \emph{Models for infrared dynamics. I. Classical currents.} Ann. Inst. H. Poincar\'e t. \textrm{26}, (1977) 109--162.


\bibitem[Ku98]{Ku98}  W. Kunhardt. \emph{On infravacua and the localization of sectors}.    J. Math. Phys. \textbf{39},  (1998) 6353.


	
	\bibitem[MS15]{MS15} G. Morchio and F. Strocchi. \emph{The infrared problem in QED: A lesson from a model with Coulomb interaction and realistic photon emission}. Preprint arXiv:1410.7289.
	


\bibitem[Po69]{Po69} K. Pohlmeyer. \emph{The Jost-Schroer theorem for 
zero-mass fields}. Commun. Math. Phys. \bf 12\rm, (1969) 204--211.

\bibitem[RS1]{RS1} M. Reed and B. Simon. \emph{Methods of modern mathematical physics. I: Functional analysis.} Academic Press 1972.
	
\bibitem[RS2]{RS2} M. Reed and B. Simon. \emph{Methods of modern mathematical physics. II: Fourier analysis, self-adjointness.} Academic Press 1975.

\bibitem[St]{St} O. Steinmann. \emph{Perturbative quantum electrodynamics and axiomatic field theory}. Springer 2000.
	
\bibitem[Ta]{Ta} M. Takesaki. \emph{Theory of operator algebras I}. Springer 1979.


\bibitem[Ta14]{Ta14} Y. Tanimoto. \emph{Massless Wigner particles in conformal field theory are free}. Forum of Mathematics, Sigma (2014), Vol 2, e21, 27 pages.


\end{thebibliography}
\end{document}